\documentclass[journal]{IEEEtran}  % Comment this line out if you need a4paper

\usepackage{subfigure}
\usepackage{amssymb, amsmath, amsfonts}
\usepackage{amsthm}
\usepackage{tikz}
\usepackage{pgfplots}
\pgfplotsset{compat=newest}
\usepackage{mathrsfs}
\usepackage{graphicx}
\usepackage{enumerate}
\usepackage[justification=centering]{caption}
\newtheorem{theorem}{Theorem}

\newtheorem{lemma}{Lemma}

\newtheorem{remark}{Remark}
\newtheorem{assumption}{Assumption}

\newlength\figureheight
\newlength\figurewidth
\newlength\fwidth

\makeatletter

\newcommand{\Rmnum}[1]{\expandafter\@slowromancap\romannumeral #1@}
\makeatother

\newcommand{\executeiffilenewer}[3]{%
	\ifnum\pdfstrcmp{\pdffilemoddate{#1}}%
	{\pdffilemoddate{#2}}>0%
	{\immediate\write18{#3}}\fi%
}

\DeclareMathOperator*{\argmin}{arg\; min}     % argmin
\DeclareMathOperator*{\argmax}{arg\; max}     % argmax
\DeclareMathOperator*{\tr}{tr}     % trace
\DeclareMathOperator{\Cov}{Cov}
\DeclareMathOperator{\rank}{rank}
\DeclareMathOperator{\diag}{diag}
\DeclareMathOperator{\logdet}{log\;det}

\title{An On-line Design of Physical Watermarks}

\author{Hanxiao Liu, Jiaqi Yan, Yilin Mo, Karl Henrik Johansson
	\thanks{
		H. Liu, J. Yan and Y. Mo are with the School of Electrical and Electronic Engineering, Nanyang Technological University, Singapore. Email: {\{hanxiao001, jyan004, ylmo\}@ntu.edu.sg}.
	}
	\thanks{
		K.H. Johansson is with the ACCESS and the Department of Automatic	Control, the School of Electrical Engineering, KTH Royal Institute of Technology, Sweden. Email: {kallej@kth.se}.
	}
	\thanks{This work is supported by Grant No. SERC A1788a0023 from ASTAR*, Singapore.}
}

\begin{document} \maketitle
\begin{abstract}
This paper considers the problem to design physical watermark signals to protect a control system against replay attacks. We first define the replay attack model, where an adversary replays the previous sensory data in order to fool the system. The physical watermarking scheme, which leverages a random control input as a watermark, to detect the replay attack is introduced. The optimal watermark signal design problem is then proposed as an optimization problem, which achieves the optimal trade-off between the control performance and attack detection performance. For the system with unknown parameters, we provide a procedure to asymptotically derive the optimal watermarking signal. Numerical examples are provided to illustrate the effectiveness of the proposed strategy.
\end{abstract}


\section{Introduction}
\IEEEPARstart{C}{yber}-Physical Systems (CPS) are the systems which offer close integration and combination between computational elements and physical processes \cite{lee2016introduction}. CPS are also defined as the system where \textquotedblleft \textit{physical and software components are deeply intertwined, each operating on different spatial and temporal scales, exhibiting multiple and distinct behavioral modalities, and interacting with each other in a myriad of ways that change with context}\textquotedblright\ \cite{bworld}. CPS play an important role in a large variety of fields, such as manufacturing, health care, environment control, transportation, military and infrastructure construction, and so on. 

Due to the wide applications and critical functions of the CPS, increasing importance has been attached to the security of CPS \cite{humayed2017cyber,sandberg2015cyberphysical}. A successful attack can jeopardize critical infrastructure and people's lives and properties, even threaten national security. In 2010, the Stuxnet malware made a devastating attack on Iranian uranium enrichment centrifuges~\cite{langner2011stuxnet}, which motivates the research community to pay more attention to the secure CPS design and defend mechanisms\cite{ani2017review} \cite{doi:10.1080/23742917.2016.1252211}. 

However, CPS security faces a wide variety of challenges. Cardenas \textit{et al.} \cite{cardenas2009challenges} discussed three main challenges that the security of CPS faces and identify the unique properties of CPS security when compared with traditional IT security. Besides, security mechanisms capable of CPS are also analyzed and some new challenges based on the physical model of process control. Similar discussion can be found in \cite{neuman2009challenges}. Gollmann and Krotofil \cite{gollmann2016cyber} pointed out that people performing the security analysis of CPS is also a key challenge. The authors argue that it is difficult for people to expertise in cyber and physical safety and able to appreciate their limitations in their own domain. 

\subsection{Previous Work}
The research community has made significant efforts in intrusion, failure and anomaly detection to enhance CPS security in recent years. Zimmer \textit{et al.} \cite{zimmer2010time} presented three mechanisms  for time-based intrusion detection. The techniques, through bounds checking, are developed in a self-checking manner by the application and through the operating system scheduler. Mitchell and Chen \cite{mitchell2011hierarchical} proposed a hierarchical performance model and techniques for intrusion detection in CPS. They classified the modern CPS intrusion detection system techniques into two classes: detection technique and audit material and summarized advantages and disadvantages in \cite{mitchell2014survey}. Kwon \textit{et al.} \cite{kwon2013security} discussed the necessary and sufficient conditions where the attacker could implement attack without being detected, which can be employed to evaluate vulnerability degree of certain CPS. And corresponding detection and defense methodologies against stealthy deception attacks also can be developed. In \cite{pasqualetti2013attack}, the authors proposed a mathematical framework for CPS and investigated limitations of the fundamental monitoring system. Besides, centralized and distributed attack detection and identification monitors were also provided.

In this paper, we consider the detection problem of replay attack, which is motivated by the Stuxnet malware mentioned previously. In \cite{mo2009secure,mo2015physical, mo2014detecting}, a replay attack model is defined and its effect on a steady-state control system is analyzed. An algebraic condition is provided on the detectability of the replay attack and for those systems that cannot detect replay attack efficiently, a physical watermarking scheme is proposed to enable the detection of the presence of the attack, by injecting a random control signal, namely watermark signal, into the control system. However, the watermark signal will deteriorate the control performance, and therefore it is important to find the optimal trade-off between the control performance loss and the detection performance, which can be casted as an optimization problem. Similar ``watermarking'' schemes are also proposed in the literature \cite{khazraei2017new,Satchidanandan2017, khazraei2017replay}. 

Different from the previous additive watermarking schemes, a multiplicative sensor watermarking scheme is proposed in \cite{Ferrari2017}. In this scheme, each output is respectively fed to a SISO watermark generator and due to the inclusion of a watermark removing functionality, the control performance will not be sacrificed. Miao \textit{et al.} \cite{FeiMiao2013} proposed the use of non-cooperative stochastic games to design a suboptimal switching control policy that balances control performance with the intrusion detection rate for replay attacks.  Hoehn and Zhang \cite{hoehn2016detection} provided a novel technique via injecting non-regular time intervals to the system and checking signal processing to detect the replay attack. Another advantage of the proposed approach is the possibility of elegant implementation into existing control systems.  
Other replay attack detection mechanisms has been proposed in the literature \cite{Shoukry2015}.

It is worth noticing that in the majority of the aforementioned researches, the precise knowledge of the system parameters is assumed in order to design the detector and the watermarking signal. However, acquiring the parameters may be troublesome and costly. Hence, it is beneficial for the system to ``learn'' the parameters during its operation and automatically design the detector and watermarking signal in real-time. Motivated by this idea, in this paper, we propose a ``learning mechanism'' to infer the system parameters as well as a physical watermark and detector design that asymptotically converges to the optimal ones. 

\subsection{Outline}
The goal of this paper is to develop a data-driven approach to design physical watermark signals to protect a cyber-physical system with unknown parameters against replay attack. First, we consider this problem under the condition of known parameters. The physical watermarking scheme leverages a random control input as a watermark to detect replay attack. The watermark signal is designed to achieve the optimal trade-off between the control performance loss and attack detection performance. Subsequently, for the system with unknown parameters, we provide an on-line ``learning'' scheme to asymptotically derive the optimal watermarking signal.   

The main contributions of this paper are as follows:
\begin{enumerate}
	\item
	The detection problem of replay attack via ``physical watermark'' with known system parameters is discussed, and a countermeasure of designing the watermarking signal that achieves the optimal trade-off between the control performance and detection performance is proposed. 
	\item 
	To the best of our knowledge, it is the first time that the detection problem of replay attack via ``physical watermark'' with unknown system parameters is discussed.
	\item 
	An on-line ``learning'' procedure is provided to asymptotically derive the optimal watermarking signal and optimal detector. 
\end{enumerate}

The rest of paper is organized as follows. Section \Rmnum{2} formulates the problem by introducing the system as well as the attack model. The physical watermarking scheme is introduced in Section \Rmnum{3}. In Section \Rmnum{4}, we present an on-line ``learning'' scheme based on the input and output data to infer the parameters of the system and design the watermark signal and the detector based on the estimated parameters. We further prove the almost sure convergence of the watermark signal to the optimal one. In Section \Rmnum{5}, numerical example is provided to verify the effectiveness of the proposed technique. Concluding remarks are given in Section \Rmnum{6}. Some proofs of theorems are included in the appendix.

\subsection*{Notations}
$\|A\|_F$ is the Frobenius norm of an $m\times n$ matrix $A$ defined as $\|A\|_F = \sqrt{\sum_{i=1}^m\sum_{j=1}^n(A_{i,j})^2}$, where $A_{i,j}$ is the $i$th row, $j$th column element of the matrix $A$. $A\otimes B$ is the Kronecker product of matrix $A$ and $B$. $A > 0$ denotes that the matrix $A$ is positive definite.

\section{Problem Formulation}
In this section, we present the problem formulation by introducing cyber-physical system model and the replay attack model, which will be employed for the rest of this paper.

We consider an linear time-invariant system described by the following equations:
\begin{align}
  x_{k+1} &= A x_k + w_{k}, \label{eq:systemdynamic}
\end{align}	
where $x_k\in \mathbb R^{n}$ is the state vector at time $k$, and $w_{k}\in \mathbb{R}^{n}$ is the zero mean Gaussian process noise with covariance $Q > 0$.

A sensor network monitors the above system. The observation equation is given by
\begin{align}
y_{k}  &= C x_k + v_k \label{eq:sensor}, 
\end{align}	
where $y_{k}\in \mathbb{R}^{m}$ is the sensor's measurement at time $k$. $v_{k}\in \mathbb{R}^{n}$ is the zero mean Gaussian measurement noise with covariance $R > 0$. 

We assume that $w_0,w_1,\cdots$ and $v_0,v_1,\cdots$ are independent of each other. Furthermore, since Cyber Physical Systems usually operate for an extended period of time, it is assumed that the system is already in the steady state, which means that $x_0$ is a zero mean Gaussian random vector independent of the process noise and the measurement noise and with covariance $\Sigma$, where $\Sigma$ satisfies:
\begin{align}
\Sigma = A\Sigma A^T+Q.
\label{eq:Sigmadef}
\end{align}

We further make the following assumptions regarding the above system:
\begin{assumption}
  The system is strictly stable. Furthermore, $(A,C)$ is observable and if system equation \eqref{eq:systemdynamic} has input matrix $B$, $(A,B)$ is controllable, which will be used after the Section \Rmnum{3}.
\end{assumption}

\begin{remark}
	The observability and controllability assumption is without loss of generality as we can perform a Kalman decomposition \cite{chen1998linear} and only work with the observable and controllable subspace.
\end{remark}

\begin{remark}
	At first glance it appears that the stability requirement regarding the system matrix $A$ could be strictly. However, we will show that our formulation can be extended to closed-loop systems. As a result, the techniques developed in this work can be used on a closed-loop system equipped with a stabilizing controller. 
\end{remark}

Next we introduce the replay attack model. We assume that the adversary have the following capabilities:
\begin{enumerate}	
\item The attacker has access to all the real-time sensory data. In other words, it knows the sensor's measurement $y_0,\cdots, y_k$ at time $k$.
\item The attacker can modify the real sensor signals $y_k$ to arbitrary sensor signals $ y_k' $.
\end{enumerate}

Given these capabilities, the adversary can employ the following  replay attack strategy:
\begin{enumerate}
\item 
  The attacker records a sequence of sensor measurements $y_k$s from time $k_1$ to $k_1+T$, where $T$ is large enough to guarantee that the attacker can replay the sequence for an extended period of time during the attack. 
\item The attacker modifies the sensor measurements $y_k$ to the recorded signals from time $k_2$ to $k_2+T$, i.e.,
  \begin{align*}
    y_{k}' = y_{k-\Delta k},\ \forall\; k_2\leq k\leq (k_2+T), %%% revise0\leq t\leq T
  \end{align*}
  where $\Delta k = k_2 - k_1$.
\end{enumerate}

Notice that since the system is already in the steady state, both the replayed signal $y_k'$ and the real signal $y_k$ from the sensors will have exactly the same statistics. As a result, for a large class of linear systems, the replayed signal and the real signal become indistinguishable after a short transient time period. In other words, it is useless for the control system to use $\chi^2$ detector to detect anomalies. For more detailed discussion, please refer to \cite{mo2009secure}.

On the other hand, the harm of the replay attack can be devastating for CPS, as is shown by the Stuxnet worm. Hence, it is crucial to detect the presence of the attack. In the next section, we introduce a detection mechanism using ``physical watermarking''.

\section{Physical Watermarking Scheme}

This section is devoted to the detection of replay attack via physical watermarking. The main idea of physical watermarking is to inject a random noise $\phi_k$, which is called the watermark signal, to excite the system and check whether the system responds to the watermark signal in accordance to the dynamical model of the system. Specifically, it is assumed that the system equation \eqref{eq:systemdynamic} is modified to
\begin{align}
  x_{k+1} = Ax_k + B \phi_k + w_k,
\end{align}
where $\phi_k\in\mathbb R^p$ is the watermark signal applied to the system at time $k$, which is usually assumed to be i.i.d. zero mean Gaussian with covariance $U$. 

In the absence of the attack, $y_{k}$ can be represented as:
\begin{align} 
  y_{k}&=\sum_{t=0}^{k-1} CA^{t}B  \phi_{k-1-t} + \sum_{t=0}^{k-1}  CA^{t} w_{k-1-t}+v_{k} + CA^k x_0.\nonumber
\end{align}
For simplicity, we define
\begin{align*}
\gamma_{k} &\triangleq \sum_{t=0}^{k} CA^{t}B  \phi_{k-t},\\
\vartheta_k &\triangleq \sum_{t=0}^{k}  CA^{t} w_{k-t}+v_{k+1} + CA^{k+1} x_0,
\end{align*}
hence, $y_{k}$ can be rewritten in the following form:
\begin{align} 
y_{k}=\gamma_{k-1} + \vartheta_{k-1}. \label{eq6}
\end{align}
Furthermore, we define 
\begin{align*}
H_\tau \triangleq CA^\tau B,
\end{align*}
hence, it is easy to know that $\gamma_{k-1}$ is a zero mean Gaussian whose covariance converges to $\mathscr U$, where
\begin{align*}
 \mathscr U = \sum_{\tau=0}^\infty H_\tau UH_\tau^T.
\end{align*}
Similarly, we can know that $\vartheta_{k-1}$ is a zero mean Gaussian noise whose covariance converges to $\mathscr W = C\Sigma C^T+R$, where $\Sigma$ is defined in \eqref{eq:Sigmadef}.

As a result, given $\phi_0, \cdots, \phi_{k-1}$, the conditional distribution of $y_k$ converges to a Gaussian distribution with mean $\gamma_{k-1}$ and covariance $\mathscr W$.

Then let us consider the scenario where the replay attack exists, the replayed $y_k'$ can be written as
\begin{align*}
  y_{k}' = y_{k-\Delta k}= \gamma_{k-1-\Delta k} + \vartheta_{k-1-\Delta k}
\end{align*}

Now since $\Delta k$ is unknown to the system operator, it is safe to assume that given $\phi_0,\cdots, \phi_{k-1}$, $y_k'$ is zero mean Gaussian with covariance $\mathscr U+ \mathscr W$.

As a result, we can design a detector to differentiate the distribution of $y_k$ under the following two hypotheses:
\begin{enumerate}
	\item[$\mathcal H_0$:] The sensor measurement $y_{k}$ follows a Gaussian distribution $ \mathcal{N}_0(\gamma_{k-1}, \mathscr W) $.
	\item[$\mathcal H_1$:] The sensor measurement $ y_{k} $ follows a Gaussian distribution $ \mathcal{N}_1(0, \mathscr U + \mathscr W) $.
\end{enumerate}

By the Neyman-Pearson lemma \cite{scharf1991statistical}, the Neyman-Pearson detector for hypothesis $\mathcal H_0$ versus hypothesis $\mathcal H_1$ takes the following form:
\begin{theorem}\label{theorem 3_1}
	The Neyman-Pearson detector rejects $\mathcal H_0$ in favor of $\mathcal H_1$ if
	\begin{equation}\label{eq11}
		\begin{split}
		  &g(y_{k},\phi_{k-1},\phi_{k-2},\cdots)\\
		  =& \big(y_{k}- \gamma_{k-1}\big)^{{T}}  \mathscr W^{-1}\big(y_{k}- \gamma_{k-1}\big)- y_{k}^{{T}}  \left(\mathscr W+\mathscr U\right)^{-1}y_{k}\\
		  \geq& \zeta,
		\end{split}
	\end{equation}
	where $\zeta$ is a predetermined threshold which depends on the desired false alarm rate. Otherwise, hypothesis $ \mathcal H_{0} $ is accepted.
\end{theorem}
\begin{remark}
  It is worth noticing that one may take a moving horizon approach to design an detector, by considering the distribution of $y_k,\,y_{k-1},\cdots, y_{k-T}$. However, the proposed methodology in this paper can be easily extended to multiple $y_k$s case by stacking the state vector.
\end{remark}

To characterize the performance of the detector, we adopt similar expected KL-divergence metrics, as is discussed in \cite{mo2015physical}, which is given by the following theorem: 
\begin{theorem}\label{theorem 2_1}
  The expected KL divergence of distribution $ \mathcal{N}_{0} $ and $ \mathcal{N}_{1} $ is 
  \begin{equation}\label{eq28}
    \mathbb{E}\ D_{KL}\left(\mathcal{N}_{1}\|\mathcal{N}_{0} \right)=\tr\left(\mathscr U \mathscr W^{-1}  \right) - \frac{1}{2}\logdet\left( I+\mathscr U \mathscr W^{-1} \right).
  \end{equation}
  Furthermore, the expected KL divergence satisfies the inequality
  \begin{equation}\label{eq29}
    \begin{split}
      \frac{1}{2}\tr\left(\mathscr U \mathscr W^{-1}  \right)&\leq \mathbb{E}\ D_{KL}\left(\mathcal{N}_{1}\|\mathcal{N}_{0} \right)\\
      & \leq \tr\left(\mathscr U \mathscr W^{-1}  \right) - \frac{1}{2}\log\left[ 1+\tr\left(\mathscr U \mathscr W^{-1}\right) \right].
    \end{split}
  \end{equation}
\end{theorem}
\begin{proof}
  The proof is similar to the proof in \cite{mo2015physical} and hence is omitted due to space limits.
\end{proof}
\begin{remark}
	It is worth noticing that the expected KL-divergence is a convex function of $\mathscr U$ and hence $U$. However, both the upper and lower bounds of it are increasing functions of $\tr(\mathscr U\mathscr W^{-1})$. Hence, in order to optimize the detection performance, we could instead maximizing $\tr(\mathscr U\mathscr W^{-1})$, which is linear with respect to $U$.
\end{remark}

Note that although the watermark signal can enable the detection of replay attack, it also deteriorates the performance of the system to some degree. As a result, it is important to find the optimal trade-off between the control performance loss and the detection performance. In this paper, to quantify the performance loss, we use the following LQG metric:
\begin{equation}\label{eq3}
  J = \lim_{T \to +\infty} \mathbb E    \left(\frac{1}{T}\sum_{k=0}^{T-1} \begin{bmatrix}
      y_k\\
      \phi_k
    \end{bmatrix}^TX \begin{bmatrix}
      y_k\\
      \phi_k
    \end{bmatrix} \right),
\end{equation}
where 
\begin{align*}
  X  = \begin{bmatrix}
    X_{yy}&X_{y\phi}\\
    X_{\phi y}&X_{\phi\phi}\\
  \end{bmatrix} > 0.
\end{align*}

Since $y_k$ and $\phi_k$ converge to a stationary process, $J$ can be written in analytical form as
\begin{align*}
  J = \lim_{k\rightarrow}\tr\left(X\Cov\left(\begin{bmatrix}
        y_k\\
        \phi_k 
      \end{bmatrix}
  \right) \right) =\tr\left(X \begin{bmatrix}
    \mathscr W + \mathscr U & H_0U\\
    UH_0^T & U
  \end{bmatrix} \right).
\end{align*}

Therefore, $J$ is an affine function of $U$, which can be written as
\begin{align*}
  J = J_0 + \Delta J = \tr(X_{yy}\mathscr W) + \tr(XS),
\end{align*}
where $S$ as a linear function of $U$ and 
\begin{align*}
  S = \begin{bmatrix}
    \mathscr U & H_0U\\
    UH_0^T & U
  \end{bmatrix}.
\end{align*}

Therefore, in order the achieve the optimal trade-off between the control performance and detection performance, we can formulate the following optimization problem: 
\begin{align}
  U =  & \argmax_{U\geq 0} & & \tr (\mathscr U\mathscr W^{-1})\nonumber\\
       & \text{subject to} & &  \tr(XS) \leq \delta,\label{eq:opt}
\end{align}
where $\delta$ is a design parameter depending on how much control performance loss is tolerable.

An important property of the optimization problem~\eqref{eq:opt} is that the optimal solution is usually a rank-$1$ matrix, which is formalized by the following theorem:
\begin{theorem}
\label{theorem:rankone}
  The optimization problem~\eqref{eq:opt} is equivalent to
  \begin{align}
    U =  & \argmax_{U\geq 0} & & \tr (U\mathcal P)\nonumber\\
         & \text{subject to} & &  \tr(U\mathcal X) \leq \delta,\label{eq:opt2}
  \end{align}
  where 
  \begin{align}
    \mathcal P &\triangleq\sum_{\tau=0}^\infty H_\tau^T\mathscr W^{-1}H_\tau,\\
    \mathcal X &\triangleq\left(\sum_{\tau=0}^\infty H_\tau^T X_{yy}H_\tau \right)+ H_0^TX_{y\phi} + X_{\phi y}H_0+ X_{\phi \phi}.
  \end{align}
  The optimal solution to \eqref{eq:opt2} is 
	\begin{align*}
	  U = zz^T,
	\end{align*}
	where $z$ is the eigenvector corresponding to the maximum eigenvalue of the matrix $\mathcal X^{-1}\mathcal P$ and $z^T\mathcal Xz = \delta$. Furthermore, the solution is unique if $\mathcal X^{-1}\mathcal P$ has only one maximum eigenvalue.
\end{theorem}
\begin{proof}
	From the definition of $\mathscr U$, we know that
	\begin{align*}
		\tr(\mathscr U\mathscr W^{-1})  &= \sum_{k=0}^\infty\tr\left(H_kUH_k^T\mathscr W^{-1}\right)\\
		&= \sum_{\tau=0}^\infty\tr\left(UH_\tau^T\mathscr W^{-1}H_\tau\right)= \tr\left(U\mathcal P\right).
	\end{align*}
	Following similar steps as in the above proof, we have that $\tr(XS) = \tr(U\mathcal X)$. Moreover, since $X>0$, we have that $\mathcal{X}>0$.
	
	If the optimal $U$ has rank greater than $1$, then follow the same line of argument in the proof of Theorem~7 in \cite{mo2014detecting}, $U$ can be decomposed as $U = \alpha_1 U_1+\cdots \alpha_l U_l$, where the following holds
	\begin{enumerate}
		\item $\alpha_i > 0$, $\sum_{i=1}^l \alpha_i = 1$.
		\item $U_i \geq 0$ is of rank $1$ and $\tr(U_i\mathcal X) = \delta$.
	\end{enumerate}
	
	Therefore, by the optimality of $U$, we can conclude that
	\begin{align*}
		\tr(U\mathcal P) = \tr(U_1\mathcal P) = \cdots = \tr(U_l\mathcal P),
	\end{align*}
	which shows that the rank one matrix $U_i$ is also optimal.
	
	For optimal rank one $U$, it can be written as $ zz^T$ for some $z \neq 0$. Hence, the optimization problem~\eqref{eq:opt2} is converted to
	\begin{align*}
		z =&  \argmax_{z\neq 0} & & z^T\mathcal P z\\
		&\text{subject to} & &z^T\mathcal Xz \leq \delta.
	\end{align*}
	Using the Lagrangian multipliers, one can prove that $\mathcal Pz = \lambda \mathcal X z$, which shows that $z$ is the eigenvector of $\mathcal X^{-1}\mathcal P$. If we enumerate all eigenvectors of $\mathcal X^{-1}\mathcal P$, it is not difficult to prove that the maximum is achieved when $z$ is the eigenvector corresponds to the largest eigenvalue of $\mathcal X^{-1}\mathcal P$ and $z^T\mathcal X z = \delta$. 
\end{proof}

Then we would like to discuss how to generalize the problem formulation for a closed-loop system with a stabilizing controller. Consider the following system discussed in \cite{mo2009secure}: 
\begin{align*}
	x_{k+1} &= A x_k + B (u_k + \phi_k) + w_k,\\
	y_k &= Cx_k + v_k, 
\end{align*}
with the following controller:
\begin{align*}
  \hat x_{k+1} &= A\hat x_k + K(y_{k+1}-CA\hat x_k),\\
  u_k &= L \hat x_k,
\end{align*}
and Linear Quadratic Gaussian (LQG) cost of
\begin{align*}
  J = \lim_{T\to\infty} \frac{1}{T}\mathbb E \left [\sum_{k=0}^{T-1} y_k^T X_{yy} y_k + (u_k+\phi_k)^TX_{\phi\phi}(u_k+\phi_k)\right],
\end{align*}
where $u_k$ denotes the optimal LQG control signal.

We can redefine the state $\tilde x_k$ and output $\tilde y_k$ as
\begin{align*}
  \tilde x_k = \begin{bmatrix}
                x_k\\
                \hat x_k
                 \end{bmatrix} \text{and}\, \tilde y_k = \begin{bmatrix}
    y_k\\
    u_k 
  \end{bmatrix},
\end{align*}
and the design of watermarking signal in a closed-loop system can be converted to the open-loop formulation.

It is worth noticing that in order to design the detector and the optimal watermarking signal, precise knowledge of the system parameters is needed. However, acquiring the parameters may be troublesome and costly. Furthermore, there may be unforeseen changes in the model of the system, such as topological changes in the power systems. Therefore, it is beneficial for the system to ``learn'' the parameters during its operation and design the detector and watermarking signal in real-time, which will be our focus in the next section.

\section{Main Results}
This section is devoted to developing an on-line ``learning'' procedure to infer the system parameters, based on which, we show how to design watermarking signals and the optimal detector and prove that the physical watermark and the detector asymptotically converges to the optimal ones. 

Throughout the section, we make the following assumptions:
\begin{assumption}\label{as2}
	\begin{enumerate}
		\item $A$ is diagonalizable and has distinct eigenvalues.
		\item The maximum eigenvalue of $\mathcal X^{-1}\mathcal P$ is unique.
		\item The system is not under attack during the ``learning'' phase.
		\item The system output $y_k$ and the dimension of the $A$ matrix $n$ is known. Furthermore, the matrix $X$ and $\delta$ are also known.
	\end{enumerate}
\end{assumption}
\begin{remark}
  Notice that the mechanism discussed in this section is similar to reinforcement learning for Markov Decision Process (MDP). However, one important distinction is that only the sensor data $y_k$ is observable and the state $x_k$ is hidden. Furthermore, we assume that we need to solve a constrained optimization problem \eqref{eq:optcharpoly} instead of an unconstrained discounted problem usually assumed in MDP.
\end{remark}

For the sake of legibility, we first introduce how to infer the necessary parameters of the system. Subsequently, we move to the design of the watermark signal and the detector based on the estimated parameters.

\subsection{Inference on the Parameters}
In this subsection, we describe the ``learning'' procedure. At each time $k$, the watermarking signal is chosen to be $\phi_k = U_k^{1/2}\zeta_k $, where $\zeta_k$s are i.i.d. Gaussian random vectors with covariance $I$. The $U_k$ is computed as a function of $y_0,\cdots, y_k,\phi_0,\cdots,\phi_{k-1}$, the procedure of which will be described in details in the next subsection. 

Consider the optimization problem~\eqref{eq:opt2}, it is easy to see that we need to infer the parameter $H_{\tau}$ and $\mathscr W$. 
Then we will show how to infer these two parameters. 
\subsection*{Inference on $H_\tau$}
This subsection is devoted to showing how to infer $H_k$. 

First, let us define the following quantity $H_{k,\tau}\ (0\leq \tau\leq 3n-2)$ as  
\begin{align*}
H_{k,\tau} &\triangleq \frac{1}{k+1} \sum_{t=0}^k y_t \phi_{t-\tau-1}^T U_{t-\tau-1}^{-1}.
\end{align*}
We assume that $\phi_{t-\tau-1} = 0$ if $t-\tau-1 <0$. One can think $H_{k,\tau}$ is an estimate of $H_\tau$.

Then, we start to prove that the $H_{k,\tau}$ converges to $H_\tau$.
\begin{theorem}
Suppose that there exists positive definite matrices $\overline M$ and $\underline M$, such that 
\begin{align}
  \frac{1}{(k+1)^\beta}\underline M \leq U_k \leq \overline M,
  \label{eq:Ubound}
\end{align}
where $0\leq \beta<1$, then $H_{k,\tau}$ converges to $H_\tau$ almost surely.
\label{theorem:convergence}
\end{theorem}
\begin{proof}
  Please see the appendix.
\end{proof}

\begin{remark}
It is worth noticing that the only requirement on $U_k$ to ensure the convergence is that $U_k$ is a function of $y_0, y_1, \cdots, y_k$ and $\phi_0, \phi_1, \cdots, \phi_{k-1}$ and satisfies \eqref{eq:Ubound}. The detailed mechanism to generate $U_k$ is not assumed by Theorem~\ref{theorem:convergence}.
\end{remark}

It is worth noticing that we only keep a record of finitely many $H_{k,\tau}$s. However, to infer matrices $\mathscr U,\,\mathscr W,\,\mathcal P$ and $\mathcal X$, we needs to estimate $H_\tau$ for all $\tau \geq 0$. 

%Since we assume that the eigenvalues of $A$ is distinct, by Cayley-Hamilton theorem, we can prove the following lemma:
\begin{lemma}
  Assuming the matrix $A$ has distinct eigenvalues $\lambda_1,\cdots,\lambda_n$, then there exists unique $\Omega_1,\cdots,\Omega_n$, such that
  \begin{align}
    H_\tau = \sum_{i=1}^n \lambda_i^\tau \Omega_i.
  \end{align}
  \label{lemma:finitetoinf}
\end{lemma}
\begin{proof}
  From Assumption \ref{as2}, $A$ can be represented in the following form:
  \begin{align*}
  A^{\tau} ={}& P\, \diag(\lambda_1^{\tau}, \lambda_2^{\tau}, \cdots, \lambda_n^{\tau})\, P^{-1} \\
    ={}& \lambda_1^{\tau}P\, \diag(1, 0, \cdots, 0)\, P^{-1}+\lambda_2^{\tau}P\ \diag(0, 1, \cdots, 0)\, P^{-1}\\
    &+\cdots+\lambda_n^{\tau}P\, \diag(0, 0, \cdots, 1)\ P^{-1}\\
    ={}&\sum_{i=1}^n \lambda_i^{\tau} A_i,
  \end{align*}
  where $\lambda_i$ denotes the $i$th eigenvalue of $A$, and $A_i = P\, \diag(0,\cdots, 1_{i}, \cdots, 0)P^{-1}\, (1\leq i\leq n)$, i.e., only the $i$th diagonal element of $A_i$ is equal to $1$, other elements are equal to 0. 
  Hence, we have
  \begin{align*}
    H_\tau &= CA^\tau B = C(\sum_{i=1}^n \lambda_i^\tau A_i)B
     = \sum_{i=1}^n \lambda_i^\tau (C A_iB)\\
     &=\sum_{i=1}^n \lambda_i^\tau\Omega_i,
  \end{align*}
  which completes the proof.
\end{proof}

%  \begin{align*}
%    A = &P\begin{bmatrix}
%    \lambda_1 &  &  & \\ 
%    &  \lambda_2&  & \\ 
%    &  &  \ddots & \\ 
%    &  &  & \lambda_n
%    \end{bmatrix}P^{-1}\\
%     =& \lambda_1P\begin{bmatrix}
%    1 &  &  & \\ 
%    &  0&  & \\ 
%    &  &  \ddots & \\ 
%    &  &  & 0
%    \end{bmatrix}P^{-1}+\lambda_2P\begin{bmatrix}
%    0 &  &  & \\ 
%    &  1&  & \\ 
%    &  &  \ddots & \\ 
%    &  &  & 0
%    \end{bmatrix}P^{-1}\\
%    &+\cdots + \lambda_nP\begin{bmatrix}
%    0 &  &  & \\ 
%    &  0&  & \\ 
%    &  &  \ddots & \\ 
%    &  &  & 1
%    \end{bmatrix}P^{-1}\\
%    =&\lambda_1A_1+\lambda_2A_2+\cdots+\lambda_n A_n,
%  \end{align*}
%  where 
%  \begin{align*}
%    A_i = \lambda_iP\begin{bmatrix}
%    0 &  &  & \\ 
%    &  0&  & \\ 
%    &  &  i_{th} & \\ 
%    &  &  & 0
%    \end{bmatrix}P^{-1}
%  \end{align*}
%Lemma~\ref{lemma:finitetoinf} can be proved using Cayley-Hamilton Theorem on $A$, the details of which is omitted due to space limit.

By Lemma~\ref{lemma:finitetoinf} and Cayley-Hamilton theorem\footnote[1]{For more details about the Cayley-Hamilton theorem, please refer to Chen \cite{chen1998linear}.}, we could use finitely many $H_0, H_1, \cdots,H_{3n-2}$ to estimate both $\lambda_i$s and $\Omega_i$s and thus $H_\tau$ for any $\tau$. To this end, let us consider the following optimization problem:
\begin{align}
  \label{eq:optcharpoly}
  \min_{\alpha_{k,0},\cdots,\alpha_{k,n-1}}  \left\|\mathscr H_k\begin{bmatrix}
      \alpha_{k,0}I\\
      \alpha_{k,1}I\\
      \cdots\\
      \alpha_{k,n-1}I
    \end{bmatrix}
 + \begin{bmatrix}
      H_{k,n}\\
      H_{k,n+1}\\
      \cdots\\
      H_{k,3n-2}
    \end{bmatrix}
 \right\|_F,
\end{align}
where $\mathscr H_k$ is a Hankel matrix defined as
\begin{align*}
\mathscr H_k\triangleq\begin{bmatrix}
H_{k,0} & H_{k,1} &\cdots & H_{k,n-1}\\
H_{k,1} & H_{k,2} &\cdots & H_{k,n}\\
\vdots & \vdots &\ddots & \vdots \\
H_{k,2n-2} & H_{k,2n-1} &\cdots & H_{k,3n-3}\\
    \end{bmatrix}.
\end{align*}

 Let us denote the roots of the polynomial $p_k(x) = x^n+\alpha_{k,n-1}x^{n-1}+\cdots+\alpha_{k,0}$ to be $\lambda_{k,1},\cdots,\lambda_{k,n}$. Define a Vandermonde like matrix $V_k$ to be
\begin{align*}
  V_k\triangleq \begin{bmatrix}
    1 & 1 &\cdots&1\\
    \lambda_{k,1} & \lambda_{k,2} &\cdots&\lambda_{k,n}\\
    \vdots & \vdots &\ddots&\vdots\\
    \lambda^{3n-2}_{k,1} & \lambda^{3n-2}_{k,2} &\cdots&\lambda^{3n-2}_{k,n}\\
  \end{bmatrix},
\end{align*}
and
\begin{align*}
 \begin{bmatrix}
   \Omega_{k,1}\\
   \vdots\\
   \Omega_{k,n}
 \end{bmatrix} = \argmax_{\Omega_{k,i}}\left\|\left(V_k\otimes I_m \right)  \begin{bmatrix}
   \Omega_{k,1}\\
   \vdots\\
   \Omega_{k,n}
 \end{bmatrix} - \begin{bmatrix}
     H_{k,0}\\
     \cdots\\
     H_{k,3n-2}
   \end{bmatrix}\right\|.
\end{align*}

The following theorem further establishes the convergence of $\lambda_{k,i}$ (and $\Omega_{k,i}$) to $\lambda_i$ (and $\Omega_{i}$):
\begin{theorem}
  \label{theorem:continuous}
  Suppose that $A$ has distinct eigenvalues. If $H_{k,\tau}$ converges to $H_\tau$ for $0\leq \tau \leq 3n-2$, then $\lambda_{k,i}$ converges $\lambda_i$ and $\Omega_{k,i}$ converges to $\Omega_i$.
\end{theorem}

Before proving Theorem~\ref{theorem:continuous}, we need to introduce the following lemma:
\begin{lemma}
	Suppose that the vector $\varphi$ is the solution of the optimization problem
	\begin{align*}
	\varphi = \argmin_{\varphi} \|A(\theta)\varphi - b(\theta)\|_2,
	\end{align*}
	where $A(\theta)$ and $b(\theta)$ are continuous functions of $\theta$. If $A(\theta_0)$ is of full column rank at $\theta_0$, then $\varphi$ is unique and a continuous function of $\theta$ in a neighborhood of $\theta_0$.
	\label{lemma:continuous}
\end{lemma}
\begin{proof}
	If $A(\theta_0)$ has full column rank, $A^TA$ is non-singular at $\theta_0$. Therefore, in a small enough neighborhood of $\theta_0$, the solution of the optimization problem can be written as
	\begin{align*}
	\varphi = (A^TA)^{-1}A^Tb, 
	\end{align*}
	which is continuous.
\end{proof}

Now we start to prove Theorem~\ref{theorem:continuous}.
\begin{proof}
  Let us denote $\alpha_i$s as the coefficients of the characteristic polynomial $p(x) = x^n+\alpha_{n-1}x^{n-1}+\alpha_0$ of $A$. By the Cayley-Hamilton theorem, if $H_{k,\tau} = H_\tau$, and $\alpha_{k,i} = \alpha_i$, then the objective function is $0$, which proves that $\alpha_i$ is a solution to \eqref{eq:optcharpoly} when $H_{k,\tau} = H_\tau$. Furthermore, the optimal value is $0$. Now suppose that $\alpha_i'$ is solution to \eqref{eq:optcharpoly}, denote the polynomial
	\begin{align*}
      p'(A) = A^n + \alpha_{n-1}'A^{n-1} + \cdots+\alpha_0'I.
	\end{align*}
  Therefore, \eqref{eq:optcharpoly} implies that $CA^ip'(A)A^jB = 0$ for all $0\leq i+j \leq 2n-2$. As a result, one can prove that
	\begin{align*}
	  \begin{bmatrix}
		C\\
		CA\\
	    \cdots\\
     	CA^{n-1}
	  \end{bmatrix}p'(A)\begin{bmatrix}
    	B&AB&\cdots&A^{n-1}B
	  \end{bmatrix}=0.
	\end{align*}
  Since we assume the system is both observable and controllable, the observability matrix and controllability matrix are full column and row rank, respectively, which implies that $p'(A) = 0$. Therefore, $p'(A)$ coincides with the characteristic polynomial since we assume that $A$ has distinct eigenvalues. Hence, the optimal solution to \eqref{eq:optcharpoly} is unique. As a result, $\mathcal H_k$ must have full column rank, which implies that $\alpha_{k,i}$s are continuous function of $H_{k,0},\cdots, H_{k,3n-2}$ in a neighborhood of $(H_0,\cdots, H_{3n-2})$. Hence, if $H_{k,\tau}$ converges to $H_\tau$, $\alpha_{k,i}$ converges to $\alpha_i$, which further implies that $\lambda_{k,i}$ converges to $\lambda_i$. As a result, $\rank(V_k) = n$ if $\lambda_{k,i} =\lambda_i$, and
  \begin{align*}
	\rank(V_k\otimes I_m) = \rank(V_k)\times \rank(I_m) = nm,
  \end{align*}
  which implies that $V_k\otimes I_m$ is full column rank. Therefore, $\Omega_{k,i}$ are continuous function of $\lambda_{k,1},\cdots,\lambda_{k,n}$ at a neighborhood of $\lambda_1,\cdots,\lambda_n$, which implies that $\Omega_{k,i}$ converges to $\Omega_i$.
\end{proof}

\subsection*{Inference on $\mathscr W$}
This subsection is devoted showing how to infer $\mathscr W$.  

First, define $Y_k$ as
\begin{align*}
Y_k  \triangleq \frac{1}{k+1}\sum_{t=0}^k y_ty_t^T,
\end{align*}
one can think $Y_k$ is an estimate of $\mathscr W + \mathscr U$. 

Then, let us define $\mathscr U_{k,ij}$, which satisfies the following recursive equation:
\begin{align*}
  \mathscr U_{k+1,ij} = \lambda_{k,i}\lambda_{k,j} \mathscr U_{k,ij} + \Omega_i U_k\Omega_j^T,
\end{align*}
and
\begin{align*}
  \mathscr U_{k} \triangleq \sum_{i=1}^n\sum_{j=1}^n \mathscr U_{k,ij}.
\end{align*}
Moreover, define
\begin{align*}
  \mathscr W_k = Y_k -\frac{1}{k+1}\sum_{t=0}^k\mathscr U_t.
\end{align*}

The following theorem establishes the convergence of $\mathscr W_k$.
\begin{theorem}
Suppose that \eqref{eq:Ubound} holds, then $\mathscr W_k$ converges to $\mathscr W$ almost surely.
\label{theorem:Pconverge}
\end{theorem}
\begin{proof}
  Please see the appendix.
\end{proof}

Before concluding this section, let us reconsider the optimization problem \eqref{eq:opt2}. This problem involves two key parameters $\mathcal P$ and $\mathcal X$. Accordingly, we further define
\begin{align*}
  \mathcal P_k &=\sum_{\tau=0}^\infty\left(\sum_{i=1}^n\lambda_{k,i}^\tau \Omega_{k,i}\right)^T\mathscr W_k^{-1} \left(\sum_{i=1}^n\lambda_{k,i}^\tau \Omega_{k,i}\right) \\
&= \sum_{i=1}^n\sum_{j=1}^n \frac{1}{1-\lambda_{k,i}\lambda_{k,j}} \Omega_{k,i}^T \mathscr W_k^{-1}\Omega_{k,j},
\end{align*}
and
\begin{align*}
	\mathcal X_k ={}& \sum_{\tau=0}^\infty\left(\sum_{i=1}^n\lambda_{k,i}^\tau \Omega_{k,i}\right)^TX_{yy} \left(\sum_{i=1}^n\lambda_{k,i}^\tau \Omega_{k,i}\right)\\
	&+ \sum_{i=1}^n \Omega_i^T X_{y\phi}+ X_{\phi y} \sum_{i=1}^n \Omega_i+X_{\phi\phi}\\
	={}& \sum_{i=1}^n\sum_{j=1}^n \frac{1}{1-\lambda_{k,i}\lambda_{k,j}} \Omega_{k,i}^T X_{yy} \Omega_{k,j} + \sum_{i=1}^n \Omega_i^T X_{y\phi}\\
	&+ X_{\phi y} \sum_{i=1}^n \Omega_i+X_{\phi\phi}.
\end{align*}
One can think  $\mathcal P_k$ and $\mathcal X_k$ are the estimation of $P$ and $X$ at time $k$, respectively. By the convergence of the parameters $H_{k,\tau}$, $\mathscr W_k$, $\lambda_{k,i}$ and $\Omega_{k,i}$, it is easy to prove that $\mathcal P_k$ and $\mathcal X_k$ converges to $\mathcal P$ and $\mathcal X$ almost surely. 

~\\
\indent
As a result, we have successfully estimated all the parameters necessary to design the detector and watermarking signal, as long as the only condition \eqref{eq:Ubound} is met.

\subsection{Watermarking Signal and Detector Design}
In this subsection, we shall show how to generate $U_k$, such that \eqref{eq:Ubound} holds and how to design a detector to detect the replay attack. Moreover, we provide a countermeasure for solving the watermarking signal and detector design problem for high dimension system. 

Here, we provide the update equation of $U_k$:
\begin{align}
  U_{k+1} = U_{k,*} + \frac{\delta}{(k+1)^\beta} I,
  \label{eq:Ukdef}
\end{align}
where $\delta$ is defined in \eqref{eq:opt} and $U_{k,*}$ is the solution of the following optimization problem
\begin{align*}
    U_{k,*} =  & \argmax_{U\geq 0} & & \tr (U\mathcal P_k)\\
         & \text{subject to} & &  \tr(U\mathcal X_k) \leq \delta,
\end{align*}
and $0\leq\beta<1$. Then we use the following theorem to establish the boundedness and convergence of $U_k$.
\begin{theorem}\label{theorem U_k}
$U_k$ is bounded by
\begin{align}
  (k+1)^{-\beta} I\leq U_k \leq \delta (X_{\phi\phi} - X_{\phi y}X_{yy}^{-1}X_{y\phi})^{-1}.
  \label{eq:Ubound2}
\end{align}
Furthermore, if $\mathcal P_k$ converges to $\mathcal P$ and $\mathcal X_k$ converges to $\mathcal X$, then
\begin{align*}
  \lim_{k\rightarrow \infty} U_k = U,
\end{align*}
where $U$ is the solution of \eqref{eq:opt2}.
\end{theorem}

\begin{proof}
 Notice that
\begin{align*}
  \mathcal X_k  &\geq \left(\sum_{i=1}^n \Omega_i\right)^T X_{yy}\left(\sum_{i=1}^n \Omega_i\right)\\
	&+ \sum_{i=1}^n \Omega_i^T X_{y\phi}+X_{\phi y} \sum_{i=1}^n \Omega_i+X_{\phi\phi}.
\end{align*}
Hence, $\mathcal X_k \geq X_{\phi\phi} - X_{\phi y}X_{yy}^{-1}X_{y\phi}$, which implies that
\begin{align}
  \label{eq:uxlessthandelta}
 \tr(U_{k,*}\left( X_{\phi\phi} - X_{\phi y}X_{yy}^{-1}X_{y\phi}\right)) \leq \delta.
\end{align}

Note that for a positive semidefinite $X$, if $\tr(X) \leq \delta$, $X\leq \delta I$. Hence, \eqref{eq:uxlessthandelta} implies that
\begin{align*}
 U_{k,*}  \leq \delta \left( X_{\phi\phi} - X_{\phi y}X_{yy}^{-1}X_{y\phi}\right)^{-1},
\end{align*}
which proves the first inequality in \eqref{eq:Ubound2}. The second inequality can be easily proved by \eqref{eq:Ukdef}.

The convergence can be proved by noticing that $U_{k,*}$ is a continuous function of $\mathcal P_k,\,\mathcal X_k$ at a neighborhood of $\mathcal P,\mathcal X$. The detailed proof is omitted due to space limit.

Now we can establish that $U_k$ converges to the optimal $U$. Notice that there is no circular logic in our proof, as \eqref{eq:Ubound2} holds regardless of the inferred value $Y_k$ and $H_{k,\tau}$. Therefore, the convergence of $\mathcal X_k$ and $\mathcal P_k$ is guaranteed by Theorem~\ref{theorem:convergence}, \ref{theorem:continuous} and \ref{theorem:Pconverge}, which further implies the convergence of $U_k$.
\end{proof}

\begin{remark}
  It is worth noticing that during the inference on the parameter and watermark signal and detector design, as long as $U_k$ satisfies the condition \eqref{eq:Ubound}, Theorem~\ref{theorem:convergence} holds, i.e., $H_{k,\tau}$ converges to $H_\tau$. Hence, Theorem~\ref{theorem:continuous} and \ref{theorem:Pconverge} holds, i.e.,  $\lambda_{k,i}$, $\Omega_{k,i}$ and $\mathscr W_k$ converge to $\lambda_{i}$, $\Omega_{i}$ and $\mathscr W$, respectively. The convergence of these three parameters guarantees the convergence of $\mathcal P_k$ and $\mathcal X_k$, which can be used to derive that $U_k$ converges to $U$, the latter of Theorem~\ref{theorem U_k}.
\end{remark}
Notice that the second term $(k+1)^{-\beta} I$ on the RHS of \eqref{eq:Ukdef} is crucial for the ``learning'' mechanism. The reason is that $U_{k,*}$ is in general a rank $1$ matrix and hence it does not provide sufficient excitation to the system for us to identify the necessary parameters. Conceptually, the problem we discussed here is similar to the multi-arm bandit problem, where one must balance exploration and exploitation. The $(k+1)^{-\beta}I$ term can be interpreted as an ``exploration'' term, as it provide necessary excitation to the system in order for us to infer the parameters. The $U_{k,*}$ is the exploitation term, as it is optimal under our current knowledge of the system parameters.

Consider the Neyman-Pearson detector in Theorem \ref{theorem 3_1}, the detector can be designed as
\begin{align}\label{eq:g_k}
  g_k = & \left(y_{k}- \tilde{\varphi}_{k-1}\right)^{{T}}  \mathscr W_k^{-1}\left(y_{k}- \tilde{\varphi}_{k-1}\right)- y_{k}^{{T}}  \left(\mathscr W_k+\mathscr U_k\right)^{-1}y_{k},
\end{align}
where $\tilde \varphi_k = \sum_{i=1}^n \tilde \varphi_{k,i}$, with
\begin{align*}
  \tilde \varphi_{k,i} = \lambda_{k,i }\tilde \varphi_{k-1,i} + \Omega_{k,i} \phi_{k}.
\end{align*}

\subsection*{Design Technique For High Dimension Systems}
This subsection is devoted to solving the watermarking signal and detector design problem for high dimension system. It is worth noticing that as complexity of the system, such as the dimension, increases, time that ``on-line'' learning process need to spend will be longer and the required resource will correspondingly becomes more. Motivated by this idea, we propose a novel technique to design the watermarking signal and detector. 

Here, we use low-dimensional system to estimate high-dimensional system and to design corresponding physical watermark and detector. For example, for an $n$-dimension system (when the dimension of the system state is equal to $n$, we use $n$-dimension for the rest of this paper), we employ $n'(n' < n)$ to estimate this system. In other words, for the real system, $A\in \mathbb{R}^{n\times n}$, $B\in \mathbb{R}^{n\times p}$, and $C\in \mathbb{R}^{m\times n}$. For the virtual system, we assume that the system is $n'$-dimension, i.e., , $A', B', C'$ can be considered as $A'\in \mathbb{R}^{n'\times n'}$, $B'\in \mathbb{R}^{n'\times p}$, $C'\in \mathbb{R}^{m\times n'}$, respectively. Correspondingly, the update equation can be written as:
\begin{align}
	U_{k+1}' = U_{k,*}' + \frac{\delta}{(k+1)^\beta} I,
\end{align}
where $U_{k,*}'$ is the solution of the following optimization problem
\begin{align*}
	U_{k,*}' =  & \argmax_{U\geq 0} & & \tr (U\mathcal P_k')\\
	& \text{subject to} & &  \tr(U\mathcal X_k') \leq \delta.
\end{align*}
Here we write several parameters as $U_{k+1}'$, $U_{k,*}'$, $\mathcal P_k'$ and $\mathcal X_k'$ since that are different from those of the real $n$-dimension system due to the choice of technique.

\section{Simulation Result}
In this section, the performance of the proposed ``learning'' procedure is evaluated. We would like to show the effectiveness from two aspects of watermarking signal design and detection performance. At the end of this section, we verify the effectiveness of the technique proposed at the end of Section \Rmnum{4}. 

Without loss of generality, we choose $n = m = p = 2$ and $A,\,B,\,C$ are all randomly generated, with $A$ stable. It is assumed that $X$ in \eqref{eq3}, the covariance matrices  $Q$ and $R$ are equal to the identity matrix with proper dimensions. We assume that $\delta$ in \eqref{eq:opt2} is equal to $10$ and $\beta = 1/3$. 

\subsection{The watermarking signal design}
This subsection is devoted to showing the effectiveness of the developed technique from the perspective of watermarking signal designed. Fig. \ref{fig:erru} shows $\|U_k-U\|_F$ v.s. time $k$, where $U$ is the solution of the optimization problem of \eqref{eq:opt}, and $U_k$, generated through updating equation \eqref{eq:Ukdef}, is an estimation of $U$. 
\begin{figure}[h!]

	
	\caption{$\|U_k-U\|_F$}
	\label{fig:erru}
\end{figure}
From Fig. \ref{fig:erru}, it can be seen that $U_k$ converges to the optimal covariance of the watermark signal $\phi_{k} $, $U$, as time goes to infinity. 
\subsection{The detection performance}
This subsection shows the performance of the developed technique from the perspective of detection performance. $g_k$ is a function of time, and it has different forms in different scenarios: 
\begin{enumerate}
	\item For the system with known parameters:
	\begin{enumerate}
		\item During the operation without replay attack, $g_k$ is equal to \eqref{eq11}.
		\item During the operation when the replay attack occurs, $g_k$ is 
		\begin{align*}
			g_k =& \big(y_{k}'- \gamma_{k-1}\big)^{{T}}  \mathscr W^{-1}\big(y_{k}'- \gamma_{k-1}\big)\\
			&- y_{k}'^{{T}}  \left(\mathscr W+\mathscr U\right)^{-1}y_{k}'.
		\end{align*} 
	\end{enumerate} 
	\item For the system with unknown parameters, during the operation with replay attack, $g_k$ is equal to \eqref{eq:g_k}. 
\end{enumerate}

It is assumed that the adversary implements the replay attack starting from time $k = 101$. Fig.~\ref{fig:noattackvsnoraml},~\ref{fig:learningvsnoraml},~\ref{fig:learningvsnoattackvsnoraml} show $g_k$ versus time $k$ in different scenarios. Here, $g_n$ denotes the value of $g_k$ during the normal operation when the replay attack is present from time $k = 101$, $g_z$ denotes the value of $g_k$ during normal operation without replay attack, and $g$ denotes the value of $g_k$ during the operation when the parameters of the system are not available to the operator, in other words, these parameters are inferred through the ``on-line'' learning technique. It is worth noticing that time in these figures is just a piece of whole operation process, which is chosen after running long enough time.
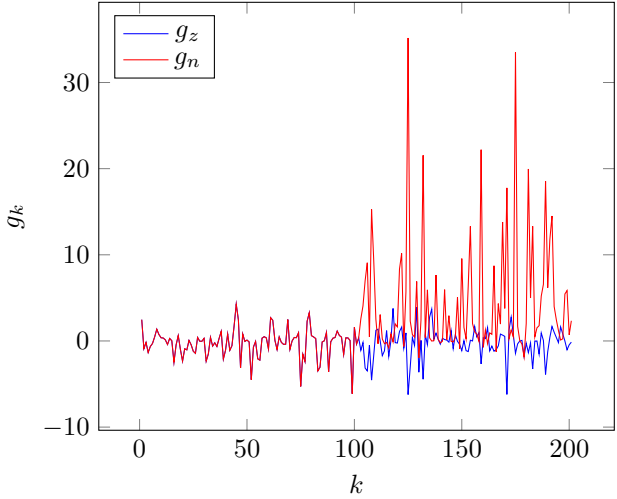
\begin{figure}[t!]
	\begin{tikzpicture}[]
	\begin{axis}[legend pos = {north west}, ylabel = {$g_k$}, xlabel = {$k$}]\addplot+ [mark = {none}, blue]coordinates {
		(1.0, 2.4709906807512216)
		(2.0, -0.8922343536710678)
		(3.0, -0.17765968701168888)
		(4.0, -1.331967248010123)
		(5.0, -0.6393628338896493)
		(6.0, -0.2973681004578772)
		(7.0, 0.45624696086663974)
		(8.0, 1.3621353707752908)
		(9.0, 0.7496459501468244)
		(10.0, 0.3716921706299623)
		(11.0, 0.3419307718702078)
		(12.0, 0.12175363976038939)
		(13.0, -0.41069237055017227)
		(14.0, 0.28901450869082185)
		(15.0, 0.02826900038914619)
		(16.0, -2.466303718682876)
		(17.0, -0.4559519834172263)
		(18.0, 0.6171183052533942)
		(19.0, -0.9771093974293655)
		(20.0, -2.34331697285408)
		(21.0, -0.9098828561742631)
		(22.0, -1.0334368253531308)
		(23.0, 0.05790827272875765)
		(24.0, -0.44008566455174764)
		(25.0, -1.196154083038347)
		(26.0, -1.4326801189548988)
		(27.0, 0.40544254351718134)
		(28.0, -0.02915170625244423)
		(29.0, -0.037830870580357834)
		(30.0, 0.297275088346006)
		(31.0, -2.3442613765191114)
		(32.0, -1.5192481041823112)
		(33.0, 0.4189845688498894)
		(34.0, -0.5944564140674142)
		(35.0, -0.21615474417533387)
		(36.0, -0.6382665518635591)
		(37.0, 0.28147106622915197)
		(38.0, 1.1041341067500774)
		(39.0, -2.106080135627428)
		(40.0, -1.0853638502200667)
		(41.0, 0.8069856609171393)
		(42.0, -1.0820761226591888)
		(43.0, -0.5517736927941654)
		(44.0, 1.7367599694456226)
		(45.0, 4.258370738497038)
		(46.0, 2.434728386350719)
		(47.0, -3.119606139529576)
		(48.0, 0.8109744731041952)
		(49.0, -0.05165201027752664)
		(50.0, 0.10204402145420399)
		(51.0, -0.11624438037349304)
		(52.0, -4.498600177902599)
		(53.0, -0.7056976584370087)
		(54.0, -0.05855191520013037)
		(55.0, -2.0761564167168034)
		(56.0, -2.2272319390465887)
		(57.0, 0.32207531062569006)
		(58.0, 0.4992534078111107)
		(59.0, 0.3572714821472727)
		(60.0, -0.930985544202261)
		(61.0, 2.700884586987166)
		(62.0, 2.399994761764147)
		(63.0, 0.06983193825434642)
		(64.0, -0.9869342880310699)
		(65.0, 0.4717313915411059)
		(66.0, -0.16256466535481895)
		(67.0, -0.3926513969111606)
		(68.0, -0.3774337173208835)
		(69.0, 2.4831927411427186)
		(70.0, -0.9653418668468609)
		(71.0, -0.0012296817107393665)
		(72.0, 0.405347895478558)
		(73.0, 0.37929573056658306)
		(74.0, 0.9513347583426148)
		(75.0, -5.273825726047544)
		(76.0, -1.5514556216203574)
		(77.0, -2.3810612956961466)
		(78.0, 2.295917173420206)
		(79.0, 3.269804014666472)
		(80.0, 0.6160604223549808)
		(81.0, 0.4684994421427058)
		(82.0, 0.27872476877332397)
		(83.0, -3.4703266090501366)
		(84.0, -3.041133130211695)
		(85.0, -0.11680403150558422)
		(86.0, 0.017680151077750537)
		(87.0, 0.982445898294394)
		(88.0, -3.58099173313707)
		(89.0, -0.09172528112412559)
		(90.0, 0.3208470866148381)
		(91.0, 0.39881469517977863)
		(92.0, 1.1470019927162292)
		(93.0, 0.6375879601724239)
		(94.0, 0.4442048300500165)
		(95.0, -1.568893802869976)
		(96.0, 0.3622254717558726)
		(97.0, 0.35428686276954147)
		(98.0, 0.05220020937548625)
		(99.0, -6.128348133933698)
		(100.0, 1.5799096903525185)
		(101.0, -0.3664514956427354)
		(102.0, 0.35688122273570555)
		(103.0, -1.087538827437406)
		(104.0, -0.19030823665679952)
		(105.0, -3.1458048098722147)
		(106.0, -3.4737447152115184)
		(107.0, -0.48914647556384083)
		(108.0, -4.545271044148997)
		(109.0, -1.6936569542637425)
		(110.0, 1.2316873994383357)
		(111.0, 1.369273217783081)
		(112.0, 0.03599975967383906)
		(113.0, -1.7012011904860662)
		(114.0, -1.1492890007633196)
		(115.0, 1.1993609568865704)
		(116.0, -1.7768726793524399)
		(117.0, 0.01839664889434811)
		(118.0, 3.7771321726988125)
		(119.0, -0.19699921418087518)
		(120.0, -0.24440979414373476)
		(121.0, 1.0928780461980332)
		(122.0, 1.6078030232595206)
		(123.0, -0.8389988238248547)
		(124.0, 0.9345571164299289)
		(125.0, -6.226159650109859)
		(126.0, -3.0199307252723444)
		(127.0, 0.38989441079236276)
		(128.0, -0.5120943765240455)
		(129.0, 3.918532579249926)
		(130.0, -3.627596744139887)
		(131.0, 0.07134778593559032)
		(132.0, -4.443124515451492)
		(133.0, 0.7253338155242757)
		(134.0, -0.4243263369469221)
		(135.0, 2.7621178513256788)
		(136.0, 3.6960226040959694)
		(137.0, 0.18984556448557488)
		(138.0, 0.9018296977393643)
		(139.0, 0.1929380924645056)
		(140.0, -0.37394205672584757)
		(141.0, 0.31182782726303904)
		(142.0, 0.16463538366978006)
		(143.0, 0.09992667960952019)
		(144.0, -0.13675757864670618)
		(145.0, 1.2231383360543333)
		(146.0, -0.8710253863669426)
		(147.0, 0.7273298710930607)
		(148.0, -0.05184505818056229)
		(149.0, -0.38933552593480186)
		(150.0, -1.1229615564137905)
		(151.0, 0.024694175708849242)
		(152.0, -1.1274679351241061)
		(153.0, -1.2248502784570119)
		(154.0, 0.08119122080717817)
		(155.0, 0.0014417897785676814)
		(156.0, 1.6739976205477998)
		(157.0, 0.4861517762980543)
		(158.0, 0.9606518091016035)
		(159.0, -2.6668344829828072)
		(160.0, 0.371579100940423)
		(161.0, 0.27189842337740017)
		(162.0, 1.530658421580421)
		(163.0, -1.1553994219212278)
		(164.0, -0.5950119518583874)
		(165.0, -1.1913338511899767)
		(166.0, -1.014058668620749)
		(167.0, -0.5711284686797788)
		(168.0, 0.7956411076599544)
		(169.0, 0.6782205954001688)
		(170.0, 0.535985387610125)
		(171.0, -6.206668049557086)
		(172.0, 1.2954978556507486)
		(173.0, 2.7609373663399674)
		(174.0, 0.7079625007032585)
		(175.0, -1.4349726664559697)
		(176.0, -0.33738639884108845)
		(177.0, -0.000767794925603249)
		(178.0, 0.06581700742389962)
		(179.0, -1.7863171912716507)
		(180.0, -0.2746071590378827)
		(181.0, -1.3484775338448172)
		(182.0, -0.2271216131416376)
		(183.0, -3.2534001707456635)
		(184.0, -0.010570005718474995)
		(185.0, -0.006983415466626619)
		(186.0, -1.5252630697274787)
		(187.0, 0.9265104787907088)
		(188.0, 0.05565623133022912)
		(189.0, -3.8853453149515413)
		(190.0, -1.2578800879318894)
		(191.0, 0.40984916551132455)
		(192.0, 1.678774249246758)
		(193.0, 1.0147558287581193)
		(194.0, 0.4112819635962852)
		(195.0, -0.19060725198047432)
		(196.0, 1.6348463644655673)
		(197.0, 0.832594243406126)
		(198.0, -0.006209906516706415)
		(199.0, -1.0214131179990484)
		(200.0, -0.45365067227065525)
		(201.0, -0.13788179686137125)
	};
	\addlegendentry{$g_z$}
	\addplot+ [mark = {none}, red]coordinates {
		(1.0, 2.4709906807512216)
		(2.0, -0.8922343536710678)
		(3.0, -0.17765968701168888)
		(4.0, -1.331967248010123)
		(5.0, -0.6393628338896493)
		(6.0, -0.2973681004578772)
		(7.0, 0.45624696086663974)
		(8.0, 1.3621353707752908)
		(9.0, 0.7496459501468244)
		(10.0, 0.3716921706299623)
		(11.0, 0.3419307718702078)
		(12.0, 0.12175363976038939)
		(13.0, -0.41069237055017227)
		(14.0, 0.28901450869082185)
		(15.0, 0.02826900038914619)
		(16.0, -2.466303718682876)
		(17.0, -0.4559519834172263)
		(18.0, 0.6171183052533942)
		(19.0, -0.9771093974293655)
		(20.0, -2.34331697285408)
		(21.0, -0.9098828561742631)
		(22.0, -1.0334368253531308)
		(23.0, 0.05790827272875765)
		(24.0, -0.44008566455174764)
		(25.0, -1.196154083038347)
		(26.0, -1.4326801189548988)
		(27.0, 0.40544254351718134)
		(28.0, -0.02915170625244423)
		(29.0, -0.037830870580357834)
		(30.0, 0.297275088346006)
		(31.0, -2.3442613765191114)
		(32.0, -1.5192481041823112)
		(33.0, 0.4189845688498894)
		(34.0, -0.5944564140674142)
		(35.0, -0.21615474417533387)
		(36.0, -0.6382665518635591)
		(37.0, 0.28147106622915197)
		(38.0, 1.1041341067500774)
		(39.0, -2.106080135627428)
		(40.0, -1.0853638502200667)
		(41.0, 0.8069856609171393)
		(42.0, -1.0820761226591888)
		(43.0, -0.5517736927941654)
		(44.0, 1.7367599694456226)
		(45.0, 4.258370738497038)
		(46.0, 2.434728386350719)
		(47.0, -3.119606139529576)
		(48.0, 0.8109744731041952)
		(49.0, -0.05165201027752664)
		(50.0, 0.10204402145420399)
		(51.0, -0.11624438037349304)
		(52.0, -4.498600177902599)
		(53.0, -0.7056976584370087)
		(54.0, -0.05855191520013037)
		(55.0, -2.0761564167168034)
		(56.0, -2.2272319390465887)
		(57.0, 0.32207531062569006)
		(58.0, 0.4992534078111107)
		(59.0, 0.3572714821472727)
		(60.0, -0.930985544202261)
		(61.0, 2.700884586987166)
		(62.0, 2.399994761764147)
		(63.0, 0.06983193825434642)
		(64.0, -0.9869342880310699)
		(65.0, 0.4717313915411059)
		(66.0, -0.16256466535481895)
		(67.0, -0.3926513969111606)
		(68.0, -0.3774337173208835)
		(69.0, 2.4831927411427186)
		(70.0, -0.9653418668468609)
		(71.0, -0.0012296817107393665)
		(72.0, 0.405347895478558)
		(73.0, 0.37929573056658306)
		(74.0, 0.9513347583426148)
		(75.0, -5.273825726047544)
		(76.0, -1.5514556216203574)
		(77.0, -2.3810612956961466)
		(78.0, 2.295917173420206)
		(79.0, 3.269804014666472)
		(80.0, 0.6160604223549808)
		(81.0, 0.4684994421427058)
		(82.0, 0.27872476877332397)
		(83.0, -3.4703266090501366)
		(84.0, -3.041133130211695)
		(85.0, -0.11680403150558422)
		(86.0, 0.017680151077750537)
		(87.0, 0.982445898294394)
		(88.0, -3.58099173313707)
		(89.0, -0.09172528112412559)
		(90.0, 0.3208470866148381)
		(91.0, 0.39881469517977863)
		(92.0, 1.1470019927162292)
		(93.0, 0.6375879601724239)
		(94.0, 0.4442048300500165)
		(95.0, -1.568893802869976)
		(96.0, 0.3622254717558726)
		(97.0, 0.35428686276954147)
		(98.0, 0.05220020937548625)
		(99.0, -6.128348133933698)
		(100.0, 1.5799096903525185)
		(101.0, -0.3664514956427354)
		(102.0, 0.35688122273570555)
		(103.0, 2.562930937133325)
		(104.0, 4.0645629058882955)
		(105.0, 6.9929981395590755)
		(106.0, 9.07907280378995)
		(107.0, 0.4629954083218122)
		(108.0, 15.26303224757036)
		(109.0, 9.386684147594146)
		(110.0, 2.5217424230042083)
		(111.0, -0.32688261739986324)
		(112.0, 3.061386549642547)
		(113.0, 0.22792563918583775)
		(114.0, -0.31827082776229276)
		(115.0, -0.16672032915112633)
		(116.0, -0.8262674909125489)
		(117.0, 0.8075298869976982)
		(118.0, -0.09957929639041829)
		(119.0, 1.9816295984005732)
		(120.0, 1.6352892854118966)
		(121.0, 8.34564063490152)
		(122.0, 10.158313453971456)
		(123.0, -0.6386711187136664)
		(124.0, 5.721700473927549)
		(125.0, 35.16867833566076)
		(126.0, 2.30613332866936)
		(127.0, 0.7434513871682737)
		(128.0, 0.45104437490047244)
		(129.0, 6.906107146425516)
		(130.0, -1.9655502810205032)
		(131.0, 2.40340530675169)
		(132.0, 21.552625947931475)
		(133.0, 0.38357692360732154)
		(134.0, 5.931905997063833)
		(135.0, 0.362016859984081)
		(136.0, -0.01335637133822365)
		(137.0, 0.03507993089462916)
		(138.0, 7.654715846134024)
		(139.0, 0.19771194011298576)
		(140.0, -0.22053865415735685)
		(141.0, -0.025266924691459858)
		(142.0, 5.9570010671000055)
		(143.0, 0.11100848835276844)
		(144.0, 2.9349485001200875)
		(145.0, -0.03632357686400589)
		(146.0, -0.11428942806503373)
		(147.0, 0.3992161928643668)
		(148.0, 5.0642984203412755)
		(149.0, -0.904874153170429)
		(150.0, 9.56459967771929)
		(151.0, 1.5792057284193457)
		(152.0, 0.24964148159798683)
		(153.0, 6.382091033834094)
		(154.0, 13.30839232833004)
		(155.0, 2.0103983705149666)
		(156.0, 1.055794486762605)
		(157.0, 0.7923595389590119)
		(158.0, -0.04675179673427221)
		(159.0, 22.20311054617408)
		(160.0, -0.7397916428401312)
		(161.0, 1.281074727415288)
		(162.0, -0.32328481611680626)
		(163.0, 0.9358072383800016)
		(164.0, 0.7728988571773685)
		(165.0, 8.73085601258397)
		(166.0, -1.261230526967518)
		(167.0, 4.338038370816348)
		(168.0, 1.959723523828791)
		(169.0, 13.807658633147382)
		(170.0, 3.777409319235236)
		(171.0, 17.76918933773433)
		(172.0, 0.2704995203648355)
		(173.0, 1.2811990406838625)
		(174.0, 0.4752083587428384)
		(175.0, 33.484187981582416)
		(176.0, 1.703226860440888)
		(177.0, 0.03884751052706084)
		(178.0, -0.3677634319969656)
		(179.0, -1.8478471563064756)
		(180.0, 2.0255731677974573)
		(181.0, 19.96280355484249)
		(182.0, 4.983698344274437)
		(183.0, 13.319564032035714)
		(184.0, 0.47147730281855105)
		(185.0, 1.5242127375511694)
		(186.0, 1.7756435729796454)
		(187.0, 5.212174162624482)
		(188.0, 6.573068729967816)
		(189.0, 18.542798076246267)
		(190.0, 6.175441044452382)
		(191.0, 11.782489952193288)
		(192.0, 14.507384915419733)
		(193.0, 4.002771937248532)
		(194.0, 2.4802420770820435)
		(195.0, 1.2162886855924238)
		(196.0, 0.10909613738349862)
		(197.0, 0.24525687861272105)
		(198.0, 5.447216011612864)
		(199.0, 5.848875511629954)
		(200.0, 0.710942827087715)
		(201.0, 2.3431985209491035)
	};
	\addlegendentry{$g_n$}
	\end{axis}
	
	\end{tikzpicture}
	
	\caption{$g_n$ and $g_z$}
	\label{fig:noattackvsnoraml} 
\end{figure}

Fig.~\ref{fig:noattackvsnoraml} shows the value of $g_n$ and $g_z$, the attacker records the sensor readings from time $1$ to $100$ and replays them to the controller from time $101$ to $200$. During the operation when replay attack is present and absent. It can be observed that there is a significant difference in the statistical distribution of $g_n$ and $g_z$ when replay attack occurs. In other words, the detector is effective.

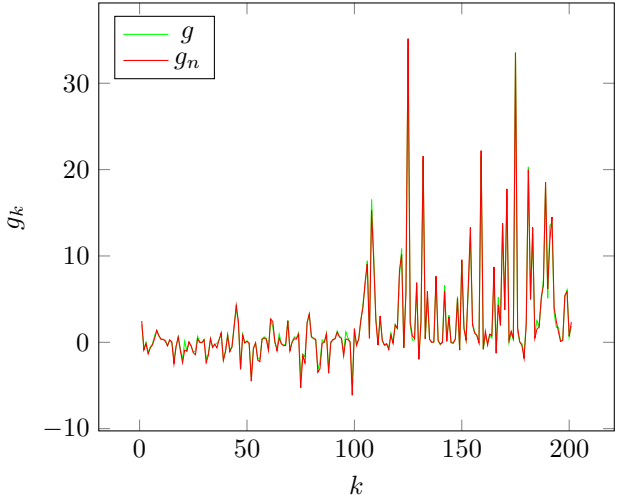
\begin{figure}[t!]
	\begin{tikzpicture}[]
	\begin{axis}[legend pos = {north west}, ylabel = {$g_k$}, xlabel = {$k$}]\addplot+ [mark = {none}, green]coordinates {
		(1.0, 2.070489596583272)
		(2.0, -0.7483949457625977)
		(3.0, 0.004364495370064558)
		(4.0, -1.188990205678348)
		(5.0, -0.5380394361645278)
		(6.0, -0.17279013395564213)
		(7.0, 0.7679894555283426)
		(8.0, 1.3985758570502902)
		(9.0, 0.7342683000158425)
		(10.0, 0.2992493083268106)
		(11.0, 0.34282574543158395)
		(12.0, 0.222042362612334)
		(13.0, -0.3690294180046747)
		(14.0, 0.2075874673414087)
		(15.0, 0.06660766355161056)
		(16.0, -1.9978282882757954)
		(17.0, -0.4838234584047668)
		(18.0, 0.6118454244179808)
		(19.0, -0.8162915058400361)
		(20.0, -1.9843831041610074)
		(21.0, 0.1304063082631317)
		(22.0, -0.9983474680338236)
		(23.0, -0.026780926648165337)
		(24.0, -0.3550343932696138)
		(25.0, -1.1576416693182)
		(26.0, -1.0084992341638626)
		(27.0, 0.6922496424539435)
		(28.0, 0.06464464075695053)
		(29.0, -0.04682188358059284)
		(30.0, 0.31432605423618876)
		(31.0, -1.899739006109694)
		(32.0, -1.3414089728271381)
		(33.0, 0.4262624423380744)
		(34.0, -0.5402743978800026)
		(35.0, -0.25469145450739306)
		(36.0, -0.13729981272926484)
		(37.0, 0.3063705788930524)
		(38.0, 0.8880703910085469)
		(39.0, -1.938325437517892)
		(40.0, -0.9813126018083862)
		(41.0, 0.9031182287772194)
		(42.0, -0.8590695933434908)
		(43.0, -0.5370599575625972)
		(44.0, 1.9774839653445504)
		(45.0, 4.014667304725821)
		(46.0, 2.070046461937201)
		(47.0, -2.9225950103859164)
		(48.0, 0.9185272573944635)
		(49.0, -0.10383908284834487)
		(50.0, 0.1868410409474779)
		(51.0, -0.10136732963894345)
		(52.0, -4.0203915438996045)
		(53.0, -0.6265382660653367)
		(54.0, -0.13190507011612373)
		(55.0, -1.8495195422969095)
		(56.0, -1.9472566403194427)
		(57.0, 0.40543288608276185)
		(58.0, 0.6294302185193308)
		(59.0, 0.4669148699824551)
		(60.0, -0.8123917760267589)
		(61.0, 2.546731219029831)
		(62.0, 1.8491382104810188)
		(63.0, -0.11991017246718139)
		(64.0, -0.6085730968266754)
		(65.0, 0.9405129305998066)
		(66.0, -0.14498297670146987)
		(67.0, -0.32027544154246146)
		(68.0, -0.22550711494158104)
		(69.0, 2.533482093699733)
		(70.0, -0.9087794555092199)
		(71.0, 0.0705035543148502)
		(72.0, 0.626183042962394)
		(73.0, 0.46299654378281513)
		(74.0, 1.0745126258269122)
		(75.0, -3.9006864222492847)
		(76.0, -1.3746939102329536)
		(77.0, -1.6537244911129285)
		(78.0, 2.260666111650777)
		(79.0, 3.2327888414403203)
		(80.0, 0.692090263451425)
		(81.0, 0.40992642771571486)
		(82.0, 0.38930017986499343)
		(83.0, -3.083353204530093)
		(84.0, -2.2656530370405994)
		(85.0, 0.2585436349791964)
		(86.0, 0.00773531600968419)
		(87.0, 1.034658586834576)
		(88.0, -3.0532364725166916)
		(89.0, 0.06004624122462021)
		(90.0, 0.2808255092011942)
		(91.0, 0.465685691943897)
		(92.0, 1.3073708865023554)
		(93.0, 0.6981883733318259)
		(94.0, 0.46344267920443594)
		(95.0, -1.2520749813768224)
		(96.0, 1.209944505376816)
		(97.0, 0.5121513498600048)
		(98.0, 0.046555390513699446)
		(99.0, -4.888259334799366)
		(100.0, 1.1510359849263776)
		(101.0, -0.22462278516586262)
		(102.0, 0.35465338799462115)
		(103.0, 2.527991258833654)
		(104.0, 4.416273292807032)
		(105.0, 6.6060348854611375)
		(106.0, 9.446947053443449)
		(107.0, 0.9214707115442313)
		(108.0, 16.522971429567693)
		(109.0, 10.001909865289191)
		(110.0, 2.7450617054274797)
		(111.0, -0.20016156102844196)
		(112.0, 2.8666661375994815)
		(113.0, 0.16104661336071013)
		(114.0, -0.37779743660943244)
		(115.0, -0.16553774836467494)
		(116.0, -0.6372572080622372)
		(117.0, 1.0682779504283375)
		(118.0, -0.04355307241181805)
		(119.0, 2.0753701576546995)
		(120.0, 1.5725076018840614)
		(121.0, 7.743508295072731)
		(122.0, 10.849131248068483)
		(123.0, -0.5363154321311714)
		(124.0, 6.027722485270798)
		(125.0, 33.9268937473201)
		(126.0, 1.9992606038924388)
		(127.0, 0.2593323405356116)
		(128.0, 0.2630598000013251)
		(129.0, 6.535101253960714)
		(130.0, -1.941509476050991)
		(131.0, 2.2978535149015316)
		(132.0, 21.11650154991381)
		(133.0, 0.5243024278800825)
		(134.0, 5.149402257631749)
		(135.0, 0.3336168954690349)
		(136.0, 0.005058978396506764)
		(137.0, 0.03682810477043533)
		(138.0, 7.478541409399204)
		(139.0, 0.17775858277215217)
		(140.0, -0.2718443307037317)
		(141.0, 0.1240354644557986)
		(142.0, 6.598341745146629)
		(143.0, 0.17884927705557452)
		(144.0, 3.148674170587451)
		(145.0, 0.001445029149053978)
		(146.0, -0.010667262302406888)
		(147.0, 0.30656966113324935)
		(148.0, 5.281594943431444)
		(149.0, -0.8926663297717965)
		(150.0, 9.488126140561349)
		(151.0, 1.8711496376821362)
		(152.0, 0.33789905778216633)
		(153.0, 5.947561569698579)
		(154.0, 12.25461882804529)
		(155.0, 2.2528400764316268)
		(156.0, 1.1585107370012415)
		(157.0, 0.8365711449889375)
		(158.0, -0.022803811871792735)
		(159.0, 21.55197663442382)
		(160.0, -0.8341554211297697)
		(161.0, 0.8596035856957895)
		(162.0, -0.0948309903201685)
		(163.0, 0.8621999151944457)
		(164.0, 0.5003971168707677)
		(165.0, 7.797608275922023)
		(166.0, -0.9792079864304262)
		(167.0, 5.213332756006685)
		(168.0, 2.0494306707082606)
		(169.0, 13.45883349381709)
		(170.0, 4.251315588786007)
		(171.0, 17.051186256191638)
		(172.0, -0.01387516996222228)
		(173.0, 1.0639202429403718)
		(174.0, 0.3429594893544379)
		(175.0, 33.59698662358039)
		(176.0, 1.2144239725945998)
		(177.0, 0.017096943838017395)
		(178.0, -0.16044994144385494)
		(179.0, -1.564583253085942)
		(180.0, 2.290480228096582)
		(181.0, 20.315010285628844)
		(182.0, 4.99294990629234)
		(183.0, 12.552590351235116)
		(184.0, 1.0096590657761357)
		(185.0, 2.5209031894505074)
		(186.0, 1.8480842504544053)
		(187.0, 4.778810831706)
		(188.0, 7.297158011895245)
		(189.0, 18.13179258971481)
		(190.0, 5.1035619780503385)
		(191.0, 13.549708119595111)
		(192.0, 13.785178208804682)
		(193.0, 3.4291581863758216)
		(194.0, 1.7425299836428287)
		(195.0, 1.7293183795594338)
		(196.0, 0.13782199917489354)
		(197.0, 0.2547916911975505)
		(198.0, 5.291513649186617)
		(199.0, 6.047435339287256)
		(200.0, 0.6243159403562608)
		(201.0, 1.9038006139765482)
	};
	\addlegendentry{$g$}
	\addplot+ [mark = {none}, red]coordinates {
		(1.0, 2.4709906807512216)
		(2.0, -0.8922343536710678)
		(3.0, -0.17765968701168888)
		(4.0, -1.331967248010123)
		(5.0, -0.6393628338896493)
		(6.0, -0.2973681004578772)
		(7.0, 0.45624696086663974)
		(8.0, 1.3621353707752908)
		(9.0, 0.7496459501468244)
		(10.0, 0.3716921706299623)
		(11.0, 0.3419307718702078)
		(12.0, 0.12175363976038939)
		(13.0, -0.41069237055017227)
		(14.0, 0.28901450869082185)
		(15.0, 0.02826900038914619)
		(16.0, -2.466303718682876)
		(17.0, -0.4559519834172263)
		(18.0, 0.6171183052533942)
		(19.0, -0.9771093974293655)
		(20.0, -2.34331697285408)
		(21.0, -0.9098828561742631)
		(22.0, -1.0334368253531308)
		(23.0, 0.05790827272875765)
		(24.0, -0.44008566455174764)
		(25.0, -1.196154083038347)
		(26.0, -1.4326801189548988)
		(27.0, 0.40544254351718134)
		(28.0, -0.02915170625244423)
		(29.0, -0.037830870580357834)
		(30.0, 0.297275088346006)
		(31.0, -2.3442613765191114)
		(32.0, -1.5192481041823112)
		(33.0, 0.4189845688498894)
		(34.0, -0.5944564140674142)
		(35.0, -0.21615474417533387)
		(36.0, -0.6382665518635591)
		(37.0, 0.28147106622915197)
		(38.0, 1.1041341067500774)
		(39.0, -2.106080135627428)
		(40.0, -1.0853638502200667)
		(41.0, 0.8069856609171393)
		(42.0, -1.0820761226591888)
		(43.0, -0.5517736927941654)
		(44.0, 1.7367599694456226)
		(45.0, 4.258370738497038)
		(46.0, 2.434728386350719)
		(47.0, -3.119606139529576)
		(48.0, 0.8109744731041952)
		(49.0, -0.05165201027752664)
		(50.0, 0.10204402145420399)
		(51.0, -0.11624438037349304)
		(52.0, -4.498600177902599)
		(53.0, -0.7056976584370087)
		(54.0, -0.05855191520013037)
		(55.0, -2.0761564167168034)
		(56.0, -2.2272319390465887)
		(57.0, 0.32207531062569006)
		(58.0, 0.4992534078111107)
		(59.0, 0.3572714821472727)
		(60.0, -0.930985544202261)
		(61.0, 2.700884586987166)
		(62.0, 2.399994761764147)
		(63.0, 0.06983193825434642)
		(64.0, -0.9869342880310699)
		(65.0, 0.4717313915411059)
		(66.0, -0.16256466535481895)
		(67.0, -0.3926513969111606)
		(68.0, -0.3774337173208835)
		(69.0, 2.4831927411427186)
		(70.0, -0.9653418668468609)
		(71.0, -0.0012296817107393665)
		(72.0, 0.405347895478558)
		(73.0, 0.37929573056658306)
		(74.0, 0.9513347583426148)
		(75.0, -5.273825726047544)
		(76.0, -1.5514556216203574)
		(77.0, -2.3810612956961466)
		(78.0, 2.295917173420206)
		(79.0, 3.269804014666472)
		(80.0, 0.6160604223549808)
		(81.0, 0.4684994421427058)
		(82.0, 0.27872476877332397)
		(83.0, -3.4703266090501366)
		(84.0, -3.041133130211695)
		(85.0, -0.11680403150558422)
		(86.0, 0.017680151077750537)
		(87.0, 0.982445898294394)
		(88.0, -3.58099173313707)
		(89.0, -0.09172528112412559)
		(90.0, 0.3208470866148381)
		(91.0, 0.39881469517977863)
		(92.0, 1.1470019927162292)
		(93.0, 0.6375879601724239)
		(94.0, 0.4442048300500165)
		(95.0, -1.568893802869976)
		(96.0, 0.3622254717558726)
		(97.0, 0.35428686276954147)
		(98.0, 0.05220020937548625)
		(99.0, -6.128348133933698)
		(100.0, 1.5799096903525185)
		(101.0, -0.3664514956427354)
		(102.0, 0.35688122273570555)
		(103.0, 2.562930937133325)
		(104.0, 4.0645629058882955)
		(105.0, 6.9929981395590755)
		(106.0, 9.07907280378995)
		(107.0, 0.4629954083218122)
		(108.0, 15.26303224757036)
		(109.0, 9.386684147594146)
		(110.0, 2.5217424230042083)
		(111.0, -0.32688261739986324)
		(112.0, 3.061386549642547)
		(113.0, 0.22792563918583775)
		(114.0, -0.31827082776229276)
		(115.0, -0.16672032915112633)
		(116.0, -0.8262674909125489)
		(117.0, 0.8075298869976982)
		(118.0, -0.09957929639041829)
		(119.0, 1.9816295984005732)
		(120.0, 1.6352892854118966)
		(121.0, 8.34564063490152)
		(122.0, 10.158313453971456)
		(123.0, -0.6386711187136664)
		(124.0, 5.721700473927549)
		(125.0, 35.16867833566076)
		(126.0, 2.30613332866936)
		(127.0, 0.7434513871682737)
		(128.0, 0.45104437490047244)
		(129.0, 6.906107146425516)
		(130.0, -1.9655502810205032)
		(131.0, 2.40340530675169)
		(132.0, 21.552625947931475)
		(133.0, 0.38357692360732154)
		(134.0, 5.931905997063833)
		(135.0, 0.362016859984081)
		(136.0, -0.01335637133822365)
		(137.0, 0.03507993089462916)
		(138.0, 7.654715846134024)
		(139.0, 0.19771194011298576)
		(140.0, -0.22053865415735685)
		(141.0, -0.025266924691459858)
		(142.0, 5.9570010671000055)
		(143.0, 0.11100848835276844)
		(144.0, 2.9349485001200875)
		(145.0, -0.03632357686400589)
		(146.0, -0.11428942806503373)
		(147.0, 0.3992161928643668)
		(148.0, 5.0642984203412755)
		(149.0, -0.904874153170429)
		(150.0, 9.56459967771929)
		(151.0, 1.5792057284193457)
		(152.0, 0.24964148159798683)
		(153.0, 6.382091033834094)
		(154.0, 13.30839232833004)
		(155.0, 2.0103983705149666)
		(156.0, 1.055794486762605)
		(157.0, 0.7923595389590119)
		(158.0, -0.04675179673427221)
		(159.0, 22.20311054617408)
		(160.0, -0.7397916428401312)
		(161.0, 1.281074727415288)
		(162.0, -0.32328481611680626)
		(163.0, 0.9358072383800016)
		(164.0, 0.7728988571773685)
		(165.0, 8.73085601258397)
		(166.0, -1.261230526967518)
		(167.0, 4.338038370816348)
		(168.0, 1.959723523828791)
		(169.0, 13.807658633147382)
		(170.0, 3.777409319235236)
		(171.0, 17.76918933773433)
		(172.0, 0.2704995203648355)
		(173.0, 1.2811990406838625)
		(174.0, 0.4752083587428384)
		(175.0, 33.484187981582416)
		(176.0, 1.703226860440888)
		(177.0, 0.03884751052706084)
		(178.0, -0.3677634319969656)
		(179.0, -1.8478471563064756)
		(180.0, 2.0255731677974573)
		(181.0, 19.96280355484249)
		(182.0, 4.983698344274437)
		(183.0, 13.319564032035714)
		(184.0, 0.47147730281855105)
		(185.0, 1.5242127375511694)
		(186.0, 1.7756435729796454)
		(187.0, 5.212174162624482)
		(188.0, 6.573068729967816)
		(189.0, 18.542798076246267)
		(190.0, 6.175441044452382)
		(191.0, 11.782489952193288)
		(192.0, 14.507384915419733)
		(193.0, 4.002771937248532)
		(194.0, 2.4802420770820435)
		(195.0, 1.2162886855924238)
		(196.0, 0.10909613738349862)
		(197.0, 0.24525687861272105)
		(198.0, 5.447216011612864)
		(199.0, 5.848875511629954)
		(200.0, 0.710942827087715)
		(201.0, 2.3431985209491035)
	};
	\addlegendentry{$g_n$}
	\end{axis}
	
	\end{tikzpicture}
	
	\caption{$g_n$ and $g $}
	\label{fig:learningvsnoraml} 
\end{figure}

Fig.~\ref{fig:learningvsnoraml} shows the value of $g_n$ and $g$. Similar to Fig.~\ref{fig:noattackvsnoraml}, the attacker records the sensor readings and replays them. It is easy to see that there is not an appreciable difference in these two scenarios. In other words, the performance of ``on-line'' learning technique is close to that of the scenario when the system parameters are available.

\begin{figure}[h!]
	\begin{tikzpicture}[]
	\begin{axis}[legend pos = {north west}, ylabel = {$g_k$}, xlabel = {$k$}]\addplot+ [mark = {none}, green]coordinates {
		(1.0, 2.070489596583272)
		(2.0, -0.7483949457625977)
		(3.0, 0.004364495370064558)
		(4.0, -1.188990205678348)
		(5.0, -0.5380394361645278)
		(6.0, -0.17279013395564213)
		(7.0, 0.7679894555283426)
		(8.0, 1.3985758570502902)
		(9.0, 0.7342683000158425)
		(10.0, 0.2992493083268106)
		(11.0, 0.34282574543158395)
		(12.0, 0.222042362612334)
		(13.0, -0.3690294180046747)
		(14.0, 0.2075874673414087)
		(15.0, 0.06660766355161056)
		(16.0, -1.9978282882757954)
		(17.0, -0.4838234584047668)
		(18.0, 0.6118454244179808)
		(19.0, -0.8162915058400361)
		(20.0, -1.9843831041610074)
		(21.0, 0.1304063082631317)
		(22.0, -0.9983474680338236)
		(23.0, -0.026780926648165337)
		(24.0, -0.3550343932696138)
		(25.0, -1.1576416693182)
		(26.0, -1.0084992341638626)
		(27.0, 0.6922496424539435)
		(28.0, 0.06464464075695053)
		(29.0, -0.04682188358059284)
		(30.0, 0.31432605423618876)
		(31.0, -1.899739006109694)
		(32.0, -1.3414089728271381)
		(33.0, 0.4262624423380744)
		(34.0, -0.5402743978800026)
		(35.0, -0.25469145450739306)
		(36.0, -0.13729981272926484)
		(37.0, 0.3063705788930524)
		(38.0, 0.8880703910085469)
		(39.0, -1.938325437517892)
		(40.0, -0.9813126018083862)
		(41.0, 0.9031182287772194)
		(42.0, -0.8590695933434908)
		(43.0, -0.5370599575625972)
		(44.0, 1.9774839653445504)
		(45.0, 4.014667304725821)
		(46.0, 2.070046461937201)
		(47.0, -2.9225950103859164)
		(48.0, 0.9185272573944635)
		(49.0, -0.10383908284834487)
		(50.0, 0.1868410409474779)
		(51.0, -0.10136732963894345)
		(52.0, -4.0203915438996045)
		(53.0, -0.6265382660653367)
		(54.0, -0.13190507011612373)
		(55.0, -1.8495195422969095)
		(56.0, -1.9472566403194427)
		(57.0, 0.40543288608276185)
		(58.0, 0.6294302185193308)
		(59.0, 0.4669148699824551)
		(60.0, -0.8123917760267589)
		(61.0, 2.546731219029831)
		(62.0, 1.8491382104810188)
		(63.0, -0.11991017246718139)
		(64.0, -0.6085730968266754)
		(65.0, 0.9405129305998066)
		(66.0, -0.14498297670146987)
		(67.0, -0.32027544154246146)
		(68.0, -0.22550711494158104)
		(69.0, 2.533482093699733)
		(70.0, -0.9087794555092199)
		(71.0, 0.0705035543148502)
		(72.0, 0.626183042962394)
		(73.0, 0.46299654378281513)
		(74.0, 1.0745126258269122)
		(75.0, -3.9006864222492847)
		(76.0, -1.3746939102329536)
		(77.0, -1.6537244911129285)
		(78.0, 2.260666111650777)
		(79.0, 3.2327888414403203)
		(80.0, 0.692090263451425)
		(81.0, 0.40992642771571486)
		(82.0, 0.38930017986499343)
		(83.0, -3.083353204530093)
		(84.0, -2.2656530370405994)
		(85.0, 0.2585436349791964)
		(86.0, 0.00773531600968419)
		(87.0, 1.034658586834576)
		(88.0, -3.0532364725166916)
		(89.0, 0.06004624122462021)
		(90.0, 0.2808255092011942)
		(91.0, 0.465685691943897)
		(92.0, 1.3073708865023554)
		(93.0, 0.6981883733318259)
		(94.0, 0.46344267920443594)
		(95.0, -1.2520749813768224)
		(96.0, 1.209944505376816)
		(97.0, 0.5121513498600048)
		(98.0, 0.046555390513699446)
		(99.0, -4.888259334799366)
		(100.0, 1.1510359849263776)
		(101.0, -0.22462278516586262)
		(102.0, 0.35465338799462115)
		(103.0, 2.527991258833654)
		(104.0, 4.416273292807032)
		(105.0, 6.6060348854611375)
		(106.0, 9.446947053443449)
		(107.0, 0.9214707115442313)
		(108.0, 16.522971429567693)
		(109.0, 10.001909865289191)
		(110.0, 2.7450617054274797)
		(111.0, -0.20016156102844196)
		(112.0, 2.8666661375994815)
		(113.0, 0.16104661336071013)
		(114.0, -0.37779743660943244)
		(115.0, -0.16553774836467494)
		(116.0, -0.6372572080622372)
		(117.0, 1.0682779504283375)
		(118.0, -0.04355307241181805)
		(119.0, 2.0753701576546995)
		(120.0, 1.5725076018840614)
		(121.0, 7.743508295072731)
		(122.0, 10.849131248068483)
		(123.0, -0.5363154321311714)
		(124.0, 6.027722485270798)
		(125.0, 33.9268937473201)
		(126.0, 1.9992606038924388)
		(127.0, 0.2593323405356116)
		(128.0, 0.2630598000013251)
		(129.0, 6.535101253960714)
		(130.0, -1.941509476050991)
		(131.0, 2.2978535149015316)
		(132.0, 21.11650154991381)
		(133.0, 0.5243024278800825)
		(134.0, 5.149402257631749)
		(135.0, 0.3336168954690349)
		(136.0, 0.005058978396506764)
		(137.0, 0.03682810477043533)
		(138.0, 7.478541409399204)
		(139.0, 0.17775858277215217)
		(140.0, -0.2718443307037317)
		(141.0, 0.1240354644557986)
		(142.0, 6.598341745146629)
		(143.0, 0.17884927705557452)
		(144.0, 3.148674170587451)
		(145.0, 0.001445029149053978)
		(146.0, -0.010667262302406888)
		(147.0, 0.30656966113324935)
		(148.0, 5.281594943431444)
		(149.0, -0.8926663297717965)
		(150.0, 9.488126140561349)
		(151.0, 1.8711496376821362)
		(152.0, 0.33789905778216633)
		(153.0, 5.947561569698579)
		(154.0, 12.25461882804529)
		(155.0, 2.2528400764316268)
		(156.0, 1.1585107370012415)
		(157.0, 0.8365711449889375)
		(158.0, -0.022803811871792735)
		(159.0, 21.55197663442382)
		(160.0, -0.8341554211297697)
		(161.0, 0.8596035856957895)
		(162.0, -0.0948309903201685)
		(163.0, 0.8621999151944457)
		(164.0, 0.5003971168707677)
		(165.0, 7.797608275922023)
		(166.0, -0.9792079864304262)
		(167.0, 5.213332756006685)
		(168.0, 2.0494306707082606)
		(169.0, 13.45883349381709)
		(170.0, 4.251315588786007)
		(171.0, 17.051186256191638)
		(172.0, -0.01387516996222228)
		(173.0, 1.0639202429403718)
		(174.0, 0.3429594893544379)
		(175.0, 33.59698662358039)
		(176.0, 1.2144239725945998)
		(177.0, 0.017096943838017395)
		(178.0, -0.16044994144385494)
		(179.0, -1.564583253085942)
		(180.0, 2.290480228096582)
		(181.0, 20.315010285628844)
		(182.0, 4.99294990629234)
		(183.0, 12.552590351235116)
		(184.0, 1.0096590657761357)
		(185.0, 2.5209031894505074)
		(186.0, 1.8480842504544053)
		(187.0, 4.778810831706)
		(188.0, 7.297158011895245)
		(189.0, 18.13179258971481)
		(190.0, 5.1035619780503385)
		(191.0, 13.549708119595111)
		(192.0, 13.785178208804682)
		(193.0, 3.4291581863758216)
		(194.0, 1.7425299836428287)
		(195.0, 1.7293183795594338)
		(196.0, 0.13782199917489354)
		(197.0, 0.2547916911975505)
		(198.0, 5.291513649186617)
		(199.0, 6.047435339287256)
		(200.0, 0.6243159403562608)
		(201.0, 1.9038006139765482)
	};
	\addlegendentry{$g$}
	\addplot+ [mark = {none}, red]coordinates {
		(1.0, 2.4709906807512216)
		(2.0, -0.8922343536710678)
		(3.0, -0.17765968701168888)
		(4.0, -1.331967248010123)
		(5.0, -0.6393628338896493)
		(6.0, -0.2973681004578772)
		(7.0, 0.45624696086663974)
		(8.0, 1.3621353707752908)
		(9.0, 0.7496459501468244)
		(10.0, 0.3716921706299623)
		(11.0, 0.3419307718702078)
		(12.0, 0.12175363976038939)
		(13.0, -0.41069237055017227)
		(14.0, 0.28901450869082185)
		(15.0, 0.02826900038914619)
		(16.0, -2.466303718682876)
		(17.0, -0.4559519834172263)
		(18.0, 0.6171183052533942)
		(19.0, -0.9771093974293655)
		(20.0, -2.34331697285408)
		(21.0, -0.9098828561742631)
		(22.0, -1.0334368253531308)
		(23.0, 0.05790827272875765)
		(24.0, -0.44008566455174764)
		(25.0, -1.196154083038347)
		(26.0, -1.4326801189548988)
		(27.0, 0.40544254351718134)
		(28.0, -0.02915170625244423)
		(29.0, -0.037830870580357834)
		(30.0, 0.297275088346006)
		(31.0, -2.3442613765191114)
		(32.0, -1.5192481041823112)
		(33.0, 0.4189845688498894)
		(34.0, -0.5944564140674142)
		(35.0, -0.21615474417533387)
		(36.0, -0.6382665518635591)
		(37.0, 0.28147106622915197)
		(38.0, 1.1041341067500774)
		(39.0, -2.106080135627428)
		(40.0, -1.0853638502200667)
		(41.0, 0.8069856609171393)
		(42.0, -1.0820761226591888)
		(43.0, -0.5517736927941654)
		(44.0, 1.7367599694456226)
		(45.0, 4.258370738497038)
		(46.0, 2.434728386350719)
		(47.0, -3.119606139529576)
		(48.0, 0.8109744731041952)
		(49.0, -0.05165201027752664)
		(50.0, 0.10204402145420399)
		(51.0, -0.11624438037349304)
		(52.0, -4.498600177902599)
		(53.0, -0.7056976584370087)
		(54.0, -0.05855191520013037)
		(55.0, -2.0761564167168034)
		(56.0, -2.2272319390465887)
		(57.0, 0.32207531062569006)
		(58.0, 0.4992534078111107)
		(59.0, 0.3572714821472727)
		(60.0, -0.930985544202261)
		(61.0, 2.700884586987166)
		(62.0, 2.399994761764147)
		(63.0, 0.06983193825434642)
		(64.0, -0.9869342880310699)
		(65.0, 0.4717313915411059)
		(66.0, -0.16256466535481895)
		(67.0, -0.3926513969111606)
		(68.0, -0.3774337173208835)
		(69.0, 2.4831927411427186)
		(70.0, -0.9653418668468609)
		(71.0, -0.0012296817107393665)
		(72.0, 0.405347895478558)
		(73.0, 0.37929573056658306)
		(74.0, 0.9513347583426148)
		(75.0, -5.273825726047544)
		(76.0, -1.5514556216203574)
		(77.0, -2.3810612956961466)
		(78.0, 2.295917173420206)
		(79.0, 3.269804014666472)
		(80.0, 0.6160604223549808)
		(81.0, 0.4684994421427058)
		(82.0, 0.27872476877332397)
		(83.0, -3.4703266090501366)
		(84.0, -3.041133130211695)
		(85.0, -0.11680403150558422)
		(86.0, 0.017680151077750537)
		(87.0, 0.982445898294394)
		(88.0, -3.58099173313707)
		(89.0, -0.09172528112412559)
		(90.0, 0.3208470866148381)
		(91.0, 0.39881469517977863)
		(92.0, 1.1470019927162292)
		(93.0, 0.6375879601724239)
		(94.0, 0.4442048300500165)
		(95.0, -1.568893802869976)
		(96.0, 0.3622254717558726)
		(97.0, 0.35428686276954147)
		(98.0, 0.05220020937548625)
		(99.0, -6.128348133933698)
		(100.0, 1.5799096903525185)
		(101.0, -0.3664514956427354)
		(102.0, 0.35688122273570555)
		(103.0, 2.562930937133325)
		(104.0, 4.0645629058882955)
		(105.0, 6.9929981395590755)
		(106.0, 9.07907280378995)
		(107.0, 0.4629954083218122)
		(108.0, 15.26303224757036)
		(109.0, 9.386684147594146)
		(110.0, 2.5217424230042083)
		(111.0, -0.32688261739986324)
		(112.0, 3.061386549642547)
		(113.0, 0.22792563918583775)
		(114.0, -0.31827082776229276)
		(115.0, -0.16672032915112633)
		(116.0, -0.8262674909125489)
		(117.0, 0.8075298869976982)
		(118.0, -0.09957929639041829)
		(119.0, 1.9816295984005732)
		(120.0, 1.6352892854118966)
		(121.0, 8.34564063490152)
		(122.0, 10.158313453971456)
		(123.0, -0.6386711187136664)
		(124.0, 5.721700473927549)
		(125.0, 35.16867833566076)
		(126.0, 2.30613332866936)
		(127.0, 0.7434513871682737)
		(128.0, 0.45104437490047244)
		(129.0, 6.906107146425516)
		(130.0, -1.9655502810205032)
		(131.0, 2.40340530675169)
		(132.0, 21.552625947931475)
		(133.0, 0.38357692360732154)
		(134.0, 5.931905997063833)
		(135.0, 0.362016859984081)
		(136.0, -0.01335637133822365)
		(137.0, 0.03507993089462916)
		(138.0, 7.654715846134024)
		(139.0, 0.19771194011298576)
		(140.0, -0.22053865415735685)
		(141.0, -0.025266924691459858)
		(142.0, 5.9570010671000055)
		(143.0, 0.11100848835276844)
		(144.0, 2.9349485001200875)
		(145.0, -0.03632357686400589)
		(146.0, -0.11428942806503373)
		(147.0, 0.3992161928643668)
		(148.0, 5.0642984203412755)
		(149.0, -0.904874153170429)
		(150.0, 9.56459967771929)
		(151.0, 1.5792057284193457)
		(152.0, 0.24964148159798683)
		(153.0, 6.382091033834094)
		(154.0, 13.30839232833004)
		(155.0, 2.0103983705149666)
		(156.0, 1.055794486762605)
		(157.0, 0.7923595389590119)
		(158.0, -0.04675179673427221)
		(159.0, 22.20311054617408)
		(160.0, -0.7397916428401312)
		(161.0, 1.281074727415288)
		(162.0, -0.32328481611680626)
		(163.0, 0.9358072383800016)
		(164.0, 0.7728988571773685)
		(165.0, 8.73085601258397)
		(166.0, -1.261230526967518)
		(167.0, 4.338038370816348)
		(168.0, 1.959723523828791)
		(169.0, 13.807658633147382)
		(170.0, 3.777409319235236)
		(171.0, 17.76918933773433)
		(172.0, 0.2704995203648355)
		(173.0, 1.2811990406838625)
		(174.0, 0.4752083587428384)
		(175.0, 33.484187981582416)
		(176.0, 1.703226860440888)
		(177.0, 0.03884751052706084)
		(178.0, -0.3677634319969656)
		(179.0, -1.8478471563064756)
		(180.0, 2.0255731677974573)
		(181.0, 19.96280355484249)
		(182.0, 4.983698344274437)
		(183.0, 13.319564032035714)
		(184.0, 0.47147730281855105)
		(185.0, 1.5242127375511694)
		(186.0, 1.7756435729796454)
		(187.0, 5.212174162624482)
		(188.0, 6.573068729967816)
		(189.0, 18.542798076246267)
		(190.0, 6.175441044452382)
		(191.0, 11.782489952193288)
		(192.0, 14.507384915419733)
		(193.0, 4.002771937248532)
		(194.0, 2.4802420770820435)
		(195.0, 1.2162886855924238)
		(196.0, 0.10909613738349862)
		(197.0, 0.24525687861272105)
		(198.0, 5.447216011612864)
		(199.0, 5.848875511629954)
		(200.0, 0.710942827087715)
		(201.0, 2.3431985209491035)
	};
	\addlegendentry{$g_n$}
	\addplot+ [mark = {none}, blue]coordinates {
		(1.0, 2.4709906807512216)
		(2.0, -0.8922343536710678)
		(3.0, -0.17765968701168888)
		(4.0, -1.331967248010123)
		(5.0, -0.6393628338896493)
		(6.0, -0.2973681004578772)
		(7.0, 0.45624696086663974)
		(8.0, 1.3621353707752908)
		(9.0, 0.7496459501468244)
		(10.0, 0.3716921706299623)
		(11.0, 0.3419307718702078)
		(12.0, 0.12175363976038939)
		(13.0, -0.41069237055017227)
		(14.0, 0.28901450869082185)
		(15.0, 0.02826900038914619)
		(16.0, -2.466303718682876)
		(17.0, -0.4559519834172263)
		(18.0, 0.6171183052533942)
		(19.0, -0.9771093974293655)
		(20.0, -2.34331697285408)
		(21.0, -0.9098828561742631)
		(22.0, -1.0334368253531308)
		(23.0, 0.05790827272875765)
		(24.0, -0.44008566455174764)
		(25.0, -1.196154083038347)
		(26.0, -1.4326801189548988)
		(27.0, 0.40544254351718134)
		(28.0, -0.02915170625244423)
		(29.0, -0.037830870580357834)
		(30.0, 0.297275088346006)
		(31.0, -2.3442613765191114)
		(32.0, -1.5192481041823112)
		(33.0, 0.4189845688498894)
		(34.0, -0.5944564140674142)
		(35.0, -0.21615474417533387)
		(36.0, -0.6382665518635591)
		(37.0, 0.28147106622915197)
		(38.0, 1.1041341067500774)
		(39.0, -2.106080135627428)
		(40.0, -1.0853638502200667)
		(41.0, 0.8069856609171393)
		(42.0, -1.0820761226591888)
		(43.0, -0.5517736927941654)
		(44.0, 1.7367599694456226)
		(45.0, 4.258370738497038)
		(46.0, 2.434728386350719)
		(47.0, -3.119606139529576)
		(48.0, 0.8109744731041952)
		(49.0, -0.05165201027752664)
		(50.0, 0.10204402145420399)
		(51.0, -0.11624438037349304)
		(52.0, -4.498600177902599)
		(53.0, -0.7056976584370087)
		(54.0, -0.05855191520013037)
		(55.0, -2.0761564167168034)
		(56.0, -2.2272319390465887)
		(57.0, 0.32207531062569006)
		(58.0, 0.4992534078111107)
		(59.0, 0.3572714821472727)
		(60.0, -0.930985544202261)
		(61.0, 2.700884586987166)
		(62.0, 2.399994761764147)
		(63.0, 0.06983193825434642)
		(64.0, -0.9869342880310699)
		(65.0, 0.4717313915411059)
		(66.0, -0.16256466535481895)
		(67.0, -0.3926513969111606)
		(68.0, -0.3774337173208835)
		(69.0, 2.4831927411427186)
		(70.0, -0.9653418668468609)
		(71.0, -0.0012296817107393665)
		(72.0, 0.405347895478558)
		(73.0, 0.37929573056658306)
		(74.0, 0.9513347583426148)
		(75.0, -5.273825726047544)
		(76.0, -1.5514556216203574)
		(77.0, -2.3810612956961466)
		(78.0, 2.295917173420206)
		(79.0, 3.269804014666472)
		(80.0, 0.6160604223549808)
		(81.0, 0.4684994421427058)
		(82.0, 0.27872476877332397)
		(83.0, -3.4703266090501366)
		(84.0, -3.041133130211695)
		(85.0, -0.11680403150558422)
		(86.0, 0.017680151077750537)
		(87.0, 0.982445898294394)
		(88.0, -3.58099173313707)
		(89.0, -0.09172528112412559)
		(90.0, 0.3208470866148381)
		(91.0, 0.39881469517977863)
		(92.0, 1.1470019927162292)
		(93.0, 0.6375879601724239)
		(94.0, 0.4442048300500165)
		(95.0, -1.568893802869976)
		(96.0, 0.3622254717558726)
		(97.0, 0.35428686276954147)
		(98.0, 0.05220020937548625)
		(99.0, -6.128348133933698)
		(100.0, 1.5799096903525185)
		(101.0, -0.3664514956427354)
		(102.0, 0.35688122273570555)
		(103.0, -1.087538827437406)
		(104.0, -0.19030823665679952)
		(105.0, -3.1458048098722147)
		(106.0, -3.4737447152115184)
		(107.0, -0.48914647556384083)
		(108.0, -4.545271044148997)
		(109.0, -1.6936569542637425)
		(110.0, 1.2316873994383357)
		(111.0, 1.369273217783081)
		(112.0, 0.03599975967383906)
		(113.0, -1.7012011904860662)
		(114.0, -1.1492890007633196)
		(115.0, 1.1993609568865704)
		(116.0, -1.7768726793524399)
		(117.0, 0.01839664889434811)
		(118.0, 3.7771321726988125)
		(119.0, -0.19699921418087518)
		(120.0, -0.24440979414373476)
		(121.0, 1.0928780461980332)
		(122.0, 1.6078030232595206)
		(123.0, -0.8389988238248547)
		(124.0, 0.9345571164299289)
		(125.0, -6.226159650109859)
		(126.0, -3.0199307252723444)
		(127.0, 0.38989441079236276)
		(128.0, -0.5120943765240455)
		(129.0, 3.918532579249926)
		(130.0, -3.627596744139887)
		(131.0, 0.07134778593559032)
		(132.0, -4.443124515451492)
		(133.0, 0.7253338155242757)
		(134.0, -0.4243263369469221)
		(135.0, 2.7621178513256788)
		(136.0, 3.6960226040959694)
		(137.0, 0.18984556448557488)
		(138.0, 0.9018296977393643)
		(139.0, 0.1929380924645056)
		(140.0, -0.37394205672584757)
		(141.0, 0.31182782726303904)
		(142.0, 0.16463538366978006)
		(143.0, 0.09992667960952019)
		(144.0, -0.13675757864670618)
		(145.0, 1.2231383360543333)
		(146.0, -0.8710253863669426)
		(147.0, 0.7273298710930607)
		(148.0, -0.05184505818056229)
		(149.0, -0.38933552593480186)
		(150.0, -1.1229615564137905)
		(151.0, 0.024694175708849242)
		(152.0, -1.1274679351241061)
		(153.0, -1.2248502784570119)
		(154.0, 0.08119122080717817)
		(155.0, 0.0014417897785676814)
		(156.0, 1.6739976205477998)
		(157.0, 0.4861517762980543)
		(158.0, 0.9606518091016035)
		(159.0, -2.6668344829828072)
		(160.0, 0.371579100940423)
		(161.0, 0.27189842337740017)
		(162.0, 1.530658421580421)
		(163.0, -1.1553994219212278)
		(164.0, -0.5950119518583874)
		(165.0, -1.1913338511899767)
		(166.0, -1.014058668620749)
		(167.0, -0.5711284686797788)
		(168.0, 0.7956411076599544)
		(169.0, 0.6782205954001688)
		(170.0, 0.535985387610125)
		(171.0, -6.206668049557086)
		(172.0, 1.2954978556507486)
		(173.0, 2.7609373663399674)
		(174.0, 0.7079625007032585)
		(175.0, -1.4349726664559697)
		(176.0, -0.33738639884108845)
		(177.0, -0.000767794925603249)
		(178.0, 0.06581700742389962)
		(179.0, -1.7863171912716507)
		(180.0, -0.2746071590378827)
		(181.0, -1.3484775338448172)
		(182.0, -0.2271216131416376)
		(183.0, -3.2534001707456635)
		(184.0, -0.010570005718474995)
		(185.0, -0.006983415466626619)
		(186.0, -1.5252630697274787)
		(187.0, 0.9265104787907088)
		(188.0, 0.05565623133022912)
		(189.0, -3.8853453149515413)
		(190.0, -1.2578800879318894)
		(191.0, 0.40984916551132455)
		(192.0, 1.678774249246758)
		(193.0, 1.0147558287581193)
		(194.0, 0.4112819635962852)
		(195.0, -0.19060725198047432)
		(196.0, 1.6348463644655673)
		(197.0, 0.832594243406126)
		(198.0, -0.006209906516706415)
		(199.0, -1.0214131179990484)
		(200.0, -0.45365067227065525)
		(201.0, -0.13788179686137125)
	};
	\addlegendentry{$g_z$}
	\end{axis}
	
	\end{tikzpicture}
	
	\caption{$g_z$, $g_n$ and $g$ }
	\label{fig:learningvsnoattackvsnoraml} 
\end{figure}
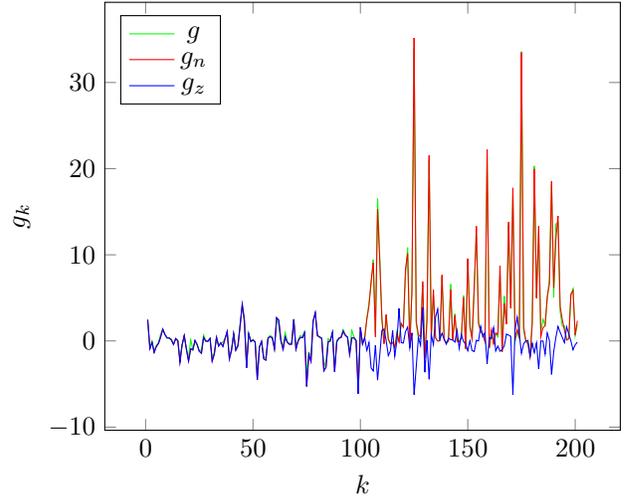

Fig.~\ref{fig:learningvsnoattackvsnoraml} shows the value of $g_n$, $g$ and $g_z$ together. We can observe that there is a significant difference in the distribution of $g$ and $g_z$, which means the ``on-line'' learning technique is effective. From this figure, one can see the performance of the ``on-line'' learning technique more clearly. 

In order to verify the effectiveness of the ``on-line'' learning technique proposed in this paper, we further employ the detection rate metric. Here, we carry out a sample set of 500 simulations to calculate the detection rate of replay attack. Here, we set $\zeta = J/0.9$, where $\zeta$ is defined in \eqref{eq11}, and $J$ is the LQG cost of the real system. Fig.~\ref{fig:detection rate} shows these curves as follows.

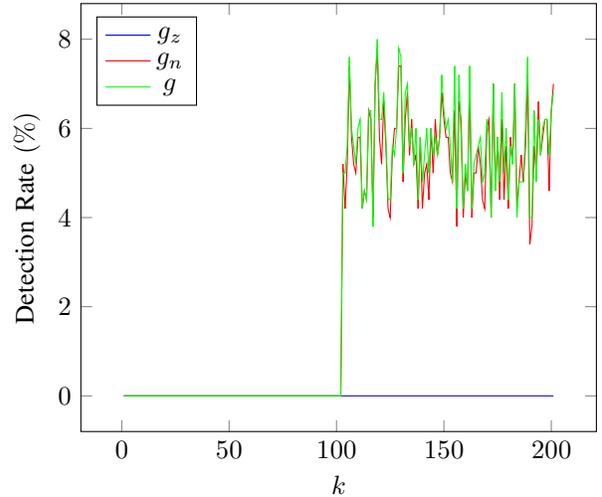
\begin{figure}[h!]
	\begin{tikzpicture}[]
	\begin{axis}[legend pos = {north west}, ylabel = {Detection Rate $(\%)$}, xlabel = {$k$}]\addplot+ [mark = {none}, blue]coordinates {
		(1.0, 0.0)
		(2.0, 0.0)
		(3.0, 0.0)
		(4.0, 0.0)
		(5.0, 0.0)
		(6.0, 0.0)
		(7.0, 0.0)
		(8.0, 0.0)
		(9.0, 0.0)
		(10.0, 0.0)
		(11.0, 0.0)
		(12.0, 0.0)
		(13.0, 0.0)
		(14.0, 0.0)
		(15.0, 0.0)
		(16.0, 0.0)
		(17.0, 0.0)
		(18.0, 0.0)
		(19.0, 0.0)
		(20.0, 0.0)
		(21.0, 0.0)
		(22.0, 0.0)
		(23.0, 0.0)
		(24.0, 0.0)
		(25.0, 0.0)
		(26.0, 0.0)
		(27.0, 0.0)
		(28.0, 0.0)
		(29.0, 0.0)
		(30.0, 0.0)
		(31.0, 0.0)
		(32.0, 0.0)
		(33.0, 0.0)
		(34.0, 0.0)
		(35.0, 0.0)
		(36.0, 0.0)
		(37.0, 0.0)
		(38.0, 0.0)
		(39.0, 0.0)
		(40.0, 0.0)
		(41.0, 0.0)
		(42.0, 0.0)
		(43.0, 0.0)
		(44.0, 0.0)
		(45.0, 0.0)
		(46.0, 0.0)
		(47.0, 0.0)
		(48.0, 0.0)
		(49.0, 0.0)
		(50.0, 0.0)
		(51.0, 0.0)
		(52.0, 0.0)
		(53.0, 0.0)
		(54.0, 0.0)
		(55.0, 0.0)
		(56.0, 0.0)
		(57.0, 0.0)
		(58.0, 0.0)
		(59.0, 0.0)
		(60.0, 0.0)
		(61.0, 0.0)
		(62.0, 0.0)
		(63.0, 0.0)
		(64.0, 0.0)
		(65.0, 0.0)
		(66.0, 0.0)
		(67.0, 0.0)
		(68.0, 0.0)
		(69.0, 0.0)
		(70.0, 0.0)
		(71.0, 0.0)
		(72.0, 0.0)
		(73.0, 0.0)
		(74.0, 0.0)
		(75.0, 0.0)
		(76.0, 0.0)
		(77.0, 0.0)
		(78.0, 0.0)
		(79.0, 0.0)
		(80.0, 0.0)
		(81.0, 0.0)
		(82.0, 0.0)
		(83.0, 0.0)
		(84.0, 0.0)
		(85.0, 0.0)
		(86.0, 0.0)
		(87.0, 0.0)
		(88.0, 0.0)
		(89.0, 0.0)
		(90.0, 0.0)
		(91.0, 0.0)
		(92.0, 0.0)
		(93.0, 0.0)
		(94.0, 0.0)
		(95.0, 0.0)
		(96.0, 0.0)
		(97.0, 0.0)
		(98.0, 0.0)
		(99.0, 0.0)
		(100.0, 0.0)
		(101.0, 0.0)
		(102.0, 0.0)
		(103.0, 0.0)
		(104.0, 0.0)
		(105.0, 0.0)
		(106.0, 0.0)
		(107.0, 0.0)
		(108.0, 0.0)
		(109.0, 0.0)
		(110.0, 0.0)
		(111.0, 0.0)
		(112.0, 0.0)
		(113.0, 0.0)
		(114.0, 0.0)
		(115.0, 0.0)
		(116.0, 0.0)
		(117.0, 0.0)
		(118.0, 0.0)
		(119.0, 0.0)
		(120.0, 0.0)
		(121.0, 0.0)
		(122.0, 0.0)
		(123.0, 0.0)
		(124.0, 0.0)
		(125.0, 0.0)
		(126.0, 0.0)
		(127.0, 0.0)
		(128.0, 0.0)
		(129.0, 0.0)
		(130.0, 0.0)
		(131.0, 0.0)
		(132.0, 0.0)
		(133.0, 0.0)
		(134.0, 0.0)
		(135.0, 0.0)
		(136.0, 0.0)
		(137.0, 0.0)
		(138.0, 0.0)
		(139.0, 0.0)
		(140.0, 0.0)
		(141.0, 0.0)
		(142.0, 0.0)
		(143.0, 0.0)
		(144.0, 0.0)
		(145.0, 0.0)
		(146.0, 0.0)
		(147.0, 0.0)
		(148.0, 0.0)
		(149.0, 0.0)
		(150.0, 0.0)
		(151.0, 0.0)
		(152.0, 0.0)
		(153.0, 0.0)
		(154.0, 0.0)
		(155.0, 0.0)
		(156.0, 0.0)
		(157.0, 0.0)
		(158.0, 0.0)
		(159.0, 0.0)
		(160.0, 0.0)
		(161.0, 0.0)
		(162.0, 0.0)
		(163.0, 0.0)
		(164.0, 0.0)
		(165.0, 0.0)
		(166.0, 0.0)
		(167.0, 0.0)
		(168.0, 0.0)
		(169.0, 0.0)
		(170.0, 0.0)
		(171.0, 0.0)
		(172.0, 0.0)
		(173.0, 0.0)
		(174.0, 0.0)
		(175.0, 0.0)
		(176.0, 0.0)
		(177.0, 0.0)
		(178.0, 0.0)
		(179.0, 0.0)
		(180.0, 0.0)
		(181.0, 0.0)
		(182.0, 0.0)
		(183.0, 0.0)
		(184.0, 0.0)
		(185.0, 0.0)
		(186.0, 0.0)
		(187.0, 0.0)
		(188.0, 0.0)
		(189.0, 0.0)
		(190.0, 0.0)
		(191.0, 0.0)
		(192.0, 0.0)
		(193.0, 0.0)
		(194.0, 0.0)
		(195.0, 0.0)
		(196.0, 0.0)
		(197.0, 0.0)
		(198.0, 0.0)
		(199.0, 0.0)
		(200.0, 0.0)
		(201.0, 0.0)
	};
	\addlegendentry{$g_z$}
	\addplot+ [mark = {none}, red]coordinates {
		(1.0, 0.0)
		(2.0, 0.0)
		(3.0, 0.0)
		(4.0, 0.0)
		(5.0, 0.0)
		(6.0, 0.0)
		(7.0, 0.0)
		(8.0, 0.0)
		(9.0, 0.0)
		(10.0, 0.0)
		(11.0, 0.0)
		(12.0, 0.0)
		(13.0, 0.0)
		(14.0, 0.0)
		(15.0, 0.0)
		(16.0, 0.0)
		(17.0, 0.0)
		(18.0, 0.0)
		(19.0, 0.0)
		(20.0, 0.0)
		(21.0, 0.0)
		(22.0, 0.0)
		(23.0, 0.0)
		(24.0, 0.0)
		(25.0, 0.0)
		(26.0, 0.0)
		(27.0, 0.0)
		(28.0, 0.0)
		(29.0, 0.0)
		(30.0, 0.0)
		(31.0, 0.0)
		(32.0, 0.0)
		(33.0, 0.0)
		(34.0, 0.0)
		(35.0, 0.0)
		(36.0, 0.0)
		(37.0, 0.0)
		(38.0, 0.0)
		(39.0, 0.0)
		(40.0, 0.0)
		(41.0, 0.0)
		(42.0, 0.0)
		(43.0, 0.0)
		(44.0, 0.0)
		(45.0, 0.0)
		(46.0, 0.0)
		(47.0, 0.0)
		(48.0, 0.0)
		(49.0, 0.0)
		(50.0, 0.0)
		(51.0, 0.0)
		(52.0, 0.0)
		(53.0, 0.0)
		(54.0, 0.0)
		(55.0, 0.0)
		(56.0, 0.0)
		(57.0, 0.0)
		(58.0, 0.0)
		(59.0, 0.0)
		(60.0, 0.0)
		(61.0, 0.0)
		(62.0, 0.0)
		(63.0, 0.0)
		(64.0, 0.0)
		(65.0, 0.0)
		(66.0, 0.0)
		(67.0, 0.0)
		(68.0, 0.0)
		(69.0, 0.0)
		(70.0, 0.0)
		(71.0, 0.0)
		(72.0, 0.0)
		(73.0, 0.0)
		(74.0, 0.0)
		(75.0, 0.0)
		(76.0, 0.0)
		(77.0, 0.0)
		(78.0, 0.0)
		(79.0, 0.0)
		(80.0, 0.0)
		(81.0, 0.0)
		(82.0, 0.0)
		(83.0, 0.0)
		(84.0, 0.0)
		(85.0, 0.0)
		(86.0, 0.0)
		(87.0, 0.0)
		(88.0, 0.0)
		(89.0, 0.0)
		(90.0, 0.0)
		(91.0, 0.0)
		(92.0, 0.0)
		(93.0, 0.0)
		(94.0, 0.0)
		(95.0, 0.0)
		(96.0, 0.0)
		(97.0, 0.0)
		(98.0, 0.0)
		(99.0, 0.0)
		(100.0, 0.0)
		(101.0, 0.0)
		(102.0, 0.0)
		(103.0, 5.2)
		(104.0, 4.2)
		(105.0, 5.2)
		(106.0, 7.3999999999999995)
		(107.0, 5.800000000000001)
		(108.0, 5.2)
		(109.0, 5.0)
		(110.0, 5.800000000000001)
		(111.0, 5.800000000000001)
		(112.0, 4.2)
		(113.0, 4.6)
		(114.0, 4.3999999999999995)
		(115.0, 6.4)
		(116.0, 6.2)
		(117.0, 4.0)
		(118.0, 7.000000000000001)
		(119.0, 7.8)
		(120.0, 5.800000000000001)
		(121.0, 5.2)
		(122.0, 6.6000000000000005)
		(123.0, 5.6000000000000005)
		(124.0, 4.2)
		(125.0, 4.0)
		(126.0, 5.4)
		(127.0, 6.0)
		(128.0, 6.0)
		(129.0, 7.3999999999999995)
		(130.0, 7.3999999999999995)
		(131.0, 4.8)
		(132.0, 6.2)
		(133.0, 6.800000000000001)
		(134.0, 5.4)
		(135.0, 6.2)
		(136.0, 5.2)
		(137.0, 5.4)
		(138.0, 4.2)
		(139.0, 5.800000000000001)
		(140.0, 4.2)
		(141.0, 5.0)
		(142.0, 5.2)
		(143.0, 4.3999999999999995)
		(144.0, 6.0)
		(145.0, 5.0)
		(146.0, 6.2)
		(147.0, 5.4)
		(148.0, 5.800000000000001)
		(149.0, 6.800000000000001)
		(150.0, 6.2)
		(151.0, 5.800000000000001)
		(152.0, 5.800000000000001)
		(153.0, 5.0)
		(154.0, 4.8)
		(155.0, 6.4)
		(156.0, 3.8)
		(157.0, 6.6000000000000005)
		(158.0, 6.2)
		(159.0, 4.0)
		(160.0, 5.0)
		(161.0, 4.6)
		(162.0, 7.000000000000001)
		(163.0, 4.0)
		(164.0, 5.0)
		(165.0, 5.0)
		(166.0, 5.6000000000000005)
		(167.0, 5.2)
		(168.0, 4.3999999999999995)
		(169.0, 4.2)
		(170.0, 6.0)
		(171.0, 6.2)
		(172.0, 4.0)
		(173.0, 7.000000000000001)
		(174.0, 4.6)
		(175.0, 5.800000000000001)
		(176.0, 4.3999999999999995)
		(177.0, 6.2)
		(178.0, 4.3999999999999995)
		(179.0, 5.6000000000000005)
		(180.0, 4.2)
		(181.0, 5.800000000000001)
		(182.0, 5.2)
		(183.0, 7.000000000000001)
		(184.0, 4.2)
		(185.0, 4.8)
		(186.0, 5.4)
		(187.0, 4.8)
		(188.0, 5.6000000000000005)
		(189.0, 7.199999999999999)
		(190.0, 3.4000000000000004)
		(191.0, 3.8)
		(192.0, 5.6000000000000005)
		(193.0, 5.0)
		(194.0, 6.6000000000000005)
		(195.0, 5.4)
		(196.0, 5.800000000000001)
		(197.0, 6.2)
		(198.0, 6.2)
		(199.0, 4.6)
		(200.0, 6.4)
		(201.0, 7.000000000000001)
	};
	\addlegendentry{$g_n$}
	\addplot+ [mark = {none}, green]coordinates {
		(1.0, 0.0)
		(2.0, 0.0)
		(3.0, 0.0)
		(4.0, 0.0)
		(5.0, 0.0)
		(6.0, 0.0)
		(7.0, 0.0)
		(8.0, 0.0)
		(9.0, 0.0)
		(10.0, 0.0)
		(11.0, 0.0)
		(12.0, 0.0)
		(13.0, 0.0)
		(14.0, 0.0)
		(15.0, 0.0)
		(16.0, 0.0)
		(17.0, 0.0)
		(18.0, 0.0)
		(19.0, 0.0)
		(20.0, 0.0)
		(21.0, 0.0)
		(22.0, 0.0)
		(23.0, 0.0)
		(24.0, 0.0)
		(25.0, 0.0)
		(26.0, 0.0)
		(27.0, 0.0)
		(28.0, 0.0)
		(29.0, 0.0)
		(30.0, 0.0)
		(31.0, 0.0)
		(32.0, 0.0)
		(33.0, 0.0)
		(34.0, 0.0)
		(35.0, 0.0)
		(36.0, 0.0)
		(37.0, 0.0)
		(38.0, 0.0)
		(39.0, 0.0)
		(40.0, 0.0)
		(41.0, 0.0)
		(42.0, 0.0)
		(43.0, 0.0)
		(44.0, 0.0)
		(45.0, 0.0)
		(46.0, 0.0)
		(47.0, 0.0)
		(48.0, 0.0)
		(49.0, 0.0)
		(50.0, 0.0)
		(51.0, 0.0)
		(52.0, 0.0)
		(53.0, 0.0)
		(54.0, 0.0)
		(55.0, 0.0)
		(56.0, 0.0)
		(57.0, 0.0)
		(58.0, 0.0)
		(59.0, 0.0)
		(60.0, 0.0)
		(61.0, 0.0)
		(62.0, 0.0)
		(63.0, 0.0)
		(64.0, 0.0)
		(65.0, 0.0)
		(66.0, 0.0)
		(67.0, 0.0)
		(68.0, 0.0)
		(69.0, 0.0)
		(70.0, 0.0)
		(71.0, 0.0)
		(72.0, 0.0)
		(73.0, 0.0)
		(74.0, 0.0)
		(75.0, 0.0)
		(76.0, 0.0)
		(77.0, 0.0)
		(78.0, 0.0)
		(79.0, 0.0)
		(80.0, 0.0)
		(81.0, 0.0)
		(82.0, 0.0)
		(83.0, 0.0)
		(84.0, 0.0)
		(85.0, 0.0)
		(86.0, 0.0)
		(87.0, 0.0)
		(88.0, 0.0)
		(89.0, 0.0)
		(90.0, 0.0)
		(91.0, 0.0)
		(92.0, 0.0)
		(93.0, 0.0)
		(94.0, 0.0)
		(95.0, 0.0)
		(96.0, 0.0)
		(97.0, 0.0)
		(98.0, 0.0)
		(99.0, 0.0)
		(100.0, 0.0)
		(101.0, 0.0)
		(102.0, 0.0)
		(103.0, 5.0)
		(104.0, 5.0)
		(105.0, 5.6000000000000005)
		(106.0, 7.6)
		(107.0, 6.0)
		(108.0, 5.6000000000000005)
		(109.0, 5.2)
		(110.0, 6.0)
		(111.0, 6.2)
		(112.0, 4.2)
		(113.0, 4.6)
		(114.0, 4.3999999999999995)
		(115.0, 6.4)
		(116.0, 6.4)
		(117.0, 3.8)
		(118.0, 7.000000000000001)
		(119.0, 8.0)
		(120.0, 6.2)
		(121.0, 6.2)
		(122.0, 6.800000000000001)
		(123.0, 5.800000000000001)
		(124.0, 4.3999999999999995)
		(125.0, 4.3999999999999995)
		(126.0, 5.6000000000000005)
		(127.0, 5.4)
		(128.0, 6.0)
		(129.0, 7.8)
		(130.0, 7.6)
		(131.0, 5.0)
		(132.0, 6.800000000000001)
		(133.0, 7.000000000000001)
		(134.0, 5.6000000000000005)
		(135.0, 6.0)
		(136.0, 5.2)
		(137.0, 6.0)
		(138.0, 4.3999999999999995)
		(139.0, 5.800000000000001)
		(140.0, 4.8)
		(141.0, 5.4)
		(142.0, 6.0)
		(143.0, 5.0)
		(144.0, 6.0)
		(145.0, 5.4)
		(146.0, 5.800000000000001)
		(147.0, 5.4)
		(148.0, 5.800000000000001)
		(149.0, 7.199999999999999)
		(150.0, 6.4)
		(151.0, 6.0)
		(152.0, 6.2)
		(153.0, 5.6000000000000005)
		(154.0, 4.8)
		(155.0, 7.3999999999999995)
		(156.0, 4.2)
		(157.0, 7.199999999999999)
		(158.0, 5.800000000000001)
		(159.0, 4.2)
		(160.0, 5.2)
		(161.0, 4.6)
		(162.0, 7.3999999999999995)
		(163.0, 4.2)
		(164.0, 5.2)
		(165.0, 5.4)
		(166.0, 5.6000000000000005)
		(167.0, 5.800000000000001)
		(168.0, 4.8)
		(169.0, 5.0)
		(170.0, 6.2)
		(171.0, 5.800000000000001)
		(172.0, 4.0)
		(173.0, 7.000000000000001)
		(174.0, 4.6)
		(175.0, 5.800000000000001)
		(176.0, 4.8)
		(177.0, 6.800000000000001)
		(178.0, 5.0)
		(179.0, 6.0)
		(180.0, 4.3999999999999995)
		(181.0, 5.6000000000000005)
		(182.0, 5.2)
		(183.0, 7.000000000000001)
		(184.0, 4.0)
		(185.0, 4.8)
		(186.0, 4.8)
		(187.0, 4.8)
		(188.0, 6.0)
		(189.0, 7.6)
		(190.0, 4.0)
		(191.0, 4.0)
		(192.0, 6.4)
		(193.0, 4.8)
		(194.0, 6.2)
		(195.0, 5.4)
		(196.0, 6.0)
		(197.0, 6.2)
		(198.0, 6.2)
		(199.0, 5.4)
		(200.0, 6.4)
		(201.0, 6.800000000000001)
	};
	\addlegendentry{$g$}
	\end{axis}
	
	\end{tikzpicture}
	
	\caption{Detection rate}
	\label{fig:detection rate} 
\end{figure}

For Fig.~\ref{fig:detection rate}, the blue curve $g_z$ shows the detection rate of the normal operation without replay attack. From this figure, the detection rate is equal to 0, which means there is no attack during all the process. The red curve $g_n$ shows the detection rate of the normal operation when replay attack occurs at time $k = 101$. It is easy to see that the detection rate is around $5.3\%$ staring from time $k = 101$. In other words, the replay attack can be detected. The green curve $g$ shows the scenario when the system parameters are not available to the operator. We use the ``on-line'' learning technique to infer corresponding parameters. One can notice that the detection rate is approximately same as that of the case with known system parameters. Therefore, this technique is effective for the detection of replay attack.

\begin{remark}
For this simulation, when the dimension of the system is $5\times 5\times 5$, i.e., $m = n =p = 5$, the performance of the proposed technique is satisfied when the number of update is in general more than 100,000 where the time update interval is $100$. Furthermore, in the case of $m=n=p=10$, the number will become larger.  
\end{remark}

For the technique proposed at the end of Section \Rmnum{5}, we assume that the real system is 100-dimension and $m = p = 5$, here, a 5-dimension system is employed to evaluate the real system.  The following figures show the performance of this technique.
Here we write legends as $g_{n,r}, g_{z,r}\, \text{and}\, g_{e}$ since they are different from those of the real $n$-dimension system because of the employed method. One can think that $g_{n,r}, g_{z,r}\, \text{and}\, g_{e}$ has corresponding meaning with $g_n, g_z\, \text{and}\, g$, respectively.

\begin{figure}[h!]

	
	\caption{$\|U_k'-U\|_F$}
	\label{fig:errUer} 
\end{figure}

\begin{figure}[h!]
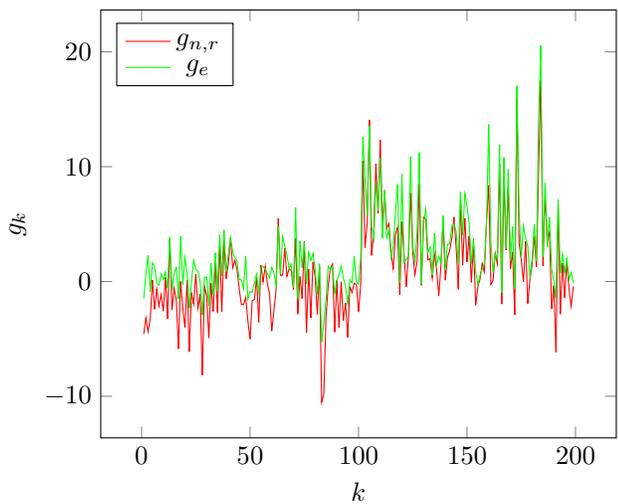

	%
	
	\caption{$g_{n,r}$ and $g_{e}$}
	\label{fig:gnrg} 
\end{figure}

\begin{figure}[h!]
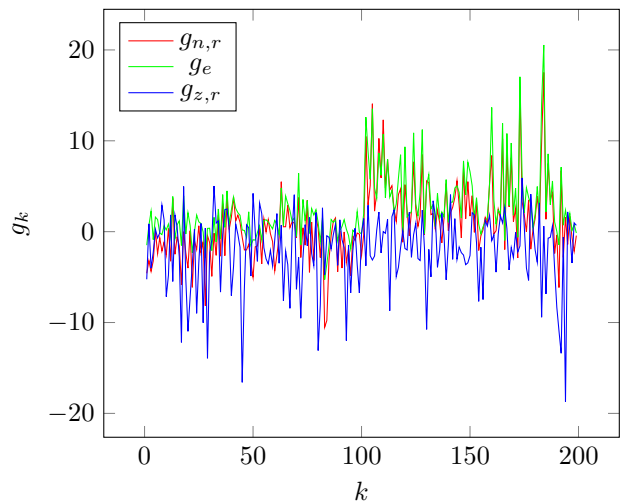

	%
	
	\caption{$g_{n,r}$, $g_e$ and $g_{z,r}$}
	\label{fig:gnrgzrg} 
\end{figure}

From Fig.~\ref{fig:errUer}, we can see that the covariance of watermarking signal converges to the optimal one designed for the real system. From Fig.~\ref{fig:gnrg} and \ref{fig:gnrgzrg}, it is easy to know that $g_k$s of real system in the presence of attack and estimation using the technique proposed with attack are close in particular when replay attack is present.  
\section{Conclusion}
In this paper, the detection problem of replay attack via ``physical watermarking'' with known system parameters is proposed to achieve the desired trade-off between the  detection performance and control performance. Then we provide an on-line ``learning'' technique for determining the optimal detector and watermarking signals without the knowledge of system parameters. The simulation is carried out to verify the effectiveness of the proposed technique. 

There are still some problems that are not covered. For example, one problem of interest is to investigate the rate of the convergence of parameters. Another interesting problem is to study how to improve the performance of the proposed technique for the high dimension case in particular when the time of computation and resources are limited. Also, how to ``learning'' parameters of a linear dynamical system in the presence of attack (The ``learning'' process is assumed without attack in this paper) might be of interest. 
\appendices
\section{Proof of Theorem~\ref{theorem:convergence}}
This section is devoted to proving Theorem~\ref{theorem:convergence}. We only prove it for the case where $\tau = 0$. The $\tau > 0$ case can be proved following similar arguments and the details are omitted due to space constraints. 

Before proving theorem~\ref{theorem:convergence}, we need to prove the following lemmas which will be used in the proof of Theorem~\ref{theorem:convergence}. Lemma~\ref{lemma:Isserlis} can be seen as an extension of Isserlis' theorem~\cite{doi:10.1093/biomet/11.3.185} to the vector case. Lemma~\ref{lemma:MatrixStrongLaw} is introduced to extend ``A strong law of large numbers for martingales''~\cite{Chow1967}, which focuses on the scalar case, to the matrix case. 

\begin{lemma}
\label{lemma:Isserlis}
Suppose that $\omega, \upsilon, \varsigma, \xi$ are four jointly Gaussian random vectors with zero mean and proper dimensions. The following equations are true:
\begin{align}
	\mathbb{E}\left[\omega \upsilon^{T} \varsigma \xi^{{T}}\right] 
	=&\mathbb{E}\left[ \omega\xi^{{T}}\right] \mathbb{E} \left[ \upsilon^{{T}} \varsigma \right]+ \mathbb{E}\left[ \omega \varsigma^{{T}} \right]  \mathbb{E} \left[\upsilon \xi^{{T}} \right]\nonumber\\
	&+\mathbb{E} \left[ \omega\upsilon^{{T}} \right]  \mathbb{E} \left[\varsigma
	\xi^{{T}} \right] ,\label{eq:vectorIsserlis}
\end{align}
and
\begin{align}
	\mathbb{E}\left[ \upsilon^{T} \varsigma \xi^{{T}}\right] = 0.
	\label{eq:isserlis3}
\end{align}
\end{lemma}
\noindent
\textit{Proof:} It is assumed that the vectors in Lemma~\ref{lemma:Isserlis} are $n$-dimension column vectors. Then let us consider the corresponding elements of the matrices on both hand sides in \eqref{eq:vectorIsserlis},
\begin{align*}
	&\mathbb{E}\left[\left[ \omega \upsilon^{T} \varsigma \xi^{T}\right]_{ij} \right] 
	=\sum_{k=1}^{n}\mathbb{E}\left[ \omega_{i} \upsilon^{T}_k  \varsigma_{k} \xi^{T}_{j}   \right]\\
	=&\sum_{k=1}^{n}\bigg(\mathbb{E}\left[ \omega_{i}\xi^{T}_{j}\right]  \mathbb{E} \left[ \upsilon^{T}_{k} \varsigma_{k}\right]+\mathbb{E}\left[ \omega_{i} \varsigma^{T}_{k} \right]  \mathbb{E} \left[\upsilon_{k} \xi^{T}_j \right]\\
	&+\mathbb{E}\left[ \omega_{i}\upsilon^{T}_k \right]  \mathbb{E} \left[\varsigma_{k}
	\xi^{T}_j \right] \bigg)\\
	={}&\mathbb{E}\left[ \omega_{i}\xi^{T}_{j}\right] \sum_{k=1}^{n} \mathbb{E} \left[ \upsilon^{T}_{k} \varsigma_{k} \right]+ \sum_{k=1}^{n}  \mathbb{E}\left[ \omega_{i} \varsigma^{T}_{k} \right]  \mathbb{E} \left[\upsilon_{k} \xi^{T}_j \right]\\
	&+  \sum_{k=1}^{n} \mathbb{E} \left[ \omega_{i}\upsilon^{T}_k \right]  \mathbb{E} \left[\varsigma_{k}
	\xi^{T}_j \right] ,\\
\end{align*}
where $ \left[ \omega \upsilon^{T} \varsigma \xi^{T}\right]_{ij}  $ denotes the $i$th row, $j$th column element of the matrix $ \left[ \omega \upsilon^{T} \varsigma \xi^{T}\right]_{ij} $, 
$ \omega_{i}$, $ \upsilon^{T}_k$, $  \varsigma_{k}$ and $\xi^{T}_{j}$ denote the $i$th, $k$th, $k$th and $j$th element of $\omega$, $\upsilon^{T}$, $  \varsigma $ and $\xi$, respectively.
Hence, the proof of \eqref{eq:vectorIsserlis} is completed.

Next let $ \omega $ to be the identity matrix $ I $ with proper dimension, we have that
\begin{align*}
	&\mathbb{E}\left[\omega \upsilon^{T} \varsigma \xi^{{T}}\right]
	= \mathbb{E}\left[I \upsilon^{T} \varsigma \xi^{{T}}\right]
	= \mathbb{E}\left[\upsilon^{T} \varsigma \xi^{{T}}\right]\\ 
	=&\mathbb{E}\left[\xi^{{T}}\right] \mathbb{E} \left[ \upsilon^{{T}} \varsigma \right]+ \mathbb{E}\left[ \varsigma^{{T}} \right]  \mathbb{E} \left[\upsilon \xi^{{T}} \right]\nonumber +  \mathbb{E} \left[\upsilon^{{T}} \right]  \mathbb{E} \left[\varsigma
	\xi^{{T}} \right]=0,
\end{align*}
which completes the proof of \eqref{eq:isserlis3}.\qed
\begin{lemma}
\label{lemma:MatrixStrongLaw}
If $\varUpsilon_{n}=\varPi_0+\varPi_1+\cdots+\varPi_n$ be a martingale such that 
\begin{align*}
\sum_{k=0}^\infty \frac{\mathbb E\left\|\varPi_k\right\|_F^2}{(k+1)^2}< \infty,
\end{align*}
where $\varPi_k(k=0,1,\cdots,n)$ and $\varUpsilon_n$ are all $m\times l$ matrices, then 
\begin{align*}
	\lim_{n\to\infty}\dfrac{\varUpsilon_n}{n+1} = \textbf{0} \ \text{almost surely}.
\end{align*}
\end{lemma}
\noindent
\textit{Proof:}
First, let us consider the martingale law of large numbers~\cite{Chow1967}, it can be rewritten as:

If $Y_n = x_1+x_2+\cdots+x_n$ is a martingale such that $\sum_{k=0}^{\infty}\mathbb{E}|x_k|^{2}/(k+1)^2<\infty$, then
$\lim_{n\to\infty} Y_n/(n+1)=0 $ almost surely.

Let $\varPi_{k,ij}$ denote the $i$th row, $j$th column element of the matrix $\varPi_{k}$, we have:
\begin{align*}
	\mathbb E\left\|\varPi_k\right\|_F^2 = 
	\mathbb{E}\left(\sum_{i=1}^{m}\sum_{j=1}^{l}|\varPi_{k,ij}|^2\right).
\end{align*}

Hence, if $ \sum_{k=0}^\infty \frac{\mathbb E\left\|\varPi_k\right\|_F^2}{(k+1)^2}< \infty$, then for arbitrary $i,j$, 
\begin{align*}
	\sum_{k=0}^\infty \frac{\mathbb{E}|\varPi_{k,ij}|^2}{(k+1)^2}<\sum_{k=0}^\infty \frac{\mathbb E\left\|\varPi_k\right\|_F^2}{(k+1)^2}< \infty,
\end{align*}
then $\lim_{n\to\infty} \varUpsilon_{n,ij}/(n+1)=0 $ almost surely. Hence, 
\begin{align*}
	\lim_{n\to\infty}\dfrac{\varUpsilon_n}{n+1} = \textbf{0} \ \text{almost surely}.
\end{align*}
which completes the proof.\qed

The proof of theorem~\ref{theorem:convergence} is as follows.

\noindent
\textit{Proof:}
Define the filtration $\mathcal F_k$ to be the $\sigma$-algebra which is generated by the following random variables $\{x_0, \phi_0,\phi_0\cdots, \phi_{k-1}, w_0,\cdots, w_{k-1},v_0,\cdots,v_k\}$. It is easy to see that both $U_k$ and $y_k$ are measurable in the $\sigma$-algebra $\mathcal F_k$. 
Let us further define 
\begin{align*}
	\mathcal S_{k} = \sum_{t=0}^k (y_{t}\phi_{t-1}^T U_{t-1}^{-1} - H_0),
\end{align*}
where $\phi_{k-1}=0$ if $k<1$. The proof is divided into steps.

First, we need to prove that $S_{k}$ is a martingale with respect to the filtration $\{\mathcal F_k\}$, i.e.,
\begin{align}
	\mathbb E (\mathcal S_{k+1}|\mathcal F_{k}) = \mathcal S_k,
	\label{eq:martingale}
\end{align}
or in other words,
\begin{align*}
	\mathbb E(y_{k+1}\phi_k^TU_k^{-1} |\mathcal F_k) = H_0.
\end{align*}

Notice that $y_{k+1}$ can be rewritten as
\begin{align*}
	y_{k+1} = \sum_{t=0}^{k} H_t   \phi_{k-t} + \sum_{t=0}^{k}  CA^{t} w_{k-t} + CA^{k+1}x_0+v_{k+1},
\end{align*}
and $\phi_k = U_k^{1/2}\zeta_k$. Therefore,
\begin{align*}
	y_{k+1}\phi_k^TU_k^{-1} =  \left(H_0 U_k^{1/2}\zeta_k  + \psi_{k+1}+v_{k+1}\right)\zeta_k^T U_k^{-1/2},
\end{align*}
where
\begin{align*}
\psi_{k+1} = \sum_{t=1}^{k} H_t   \phi_{k-t} + \sum_{t=0}^{k}  CA^{t} w_{k-t} + CA^{k+1}x_0.
\end{align*}

As it is known that $U_k$ is measurable in the $\sigma$-algebra $\mathcal F_k$ and $\psi_{k+1}$ is independent of $\zeta_k$,
\begin{align*}
	\mathbb E(\psi_{k+1}\zeta_k^TU_k^{-1/2}|\mathcal F_k) =  
	\mathbb E(\psi_{k+1}\zeta_k^T|\mathcal F_k)U_k^{-1/2} =  0.
\end{align*}
Since $v_{k+1}$ is independent of $\mathcal F_k$ and $\zeta_k$, we have
\begin{align*}
	\mathbb E(v_{k+1}\zeta_k^TU_k^{-1/2}|\mathcal F_k) =  0.
\end{align*}
Finally,
\begin{align*}
	\mathbb E( H_0 U_k^{1/2}\zeta_k \zeta_k^T U_k^{-1/2}|\mathcal F_k)  &= H_0 U_k^{1/2}\mathbb E(\zeta_k \zeta_k^T |\mathcal F_k) U_k^{-1/2} =  H_0.
\end{align*}
Therefore, \eqref{eq:martingale} holds and we establish that $\mathcal S_k$ is a martingale with respect to the filtration $\{\mathcal F_k\}$. 

Next we need to prove that
\begin{align}
	\sum_{k=0}^\infty \frac{\mathbb E\left\|y_{k+1}\phi_k^TU_k^{-1}-H_0\right\|_F^2}{(k+1)^2}< \infty.
	\label{eq:boundedsum}
\end{align}

To this end, let us consider
\begin{equation}
	\label{H0}
	\begin{aligned}
	&\left[y_{k+1}\phi_k^TU_k^{-1}-H_0\right]\left[y_{k+1}\phi_k^TU_k^{-1}-H_0\right]^T\\
	=& y_{k+1}\phi_k^TU_k^{-2}\phi_ky_{k+1}^T -H_0 U_k^{-1}\phi_ky_{k+1}^T \\
	&-  y_{k+1}\phi_k^TU_k^{-2}H_0^T + H_0 H_0^T,
	\end{aligned}
\end{equation}
where
\begin{align*}
	y_{k+1}\phi_k^TU_k^{-2}\phi_ky_{k+1}^T &= \Xi_1+ \Xi_2 + \Xi_2^T + \Xi_3,
\end{align*}
where
\begin{align*}
	\Xi_1 &= H_0U_k^{1/2}\zeta_k\zeta_k^TU_k^{-1}\zeta_k \zeta_k^TU_k^{1/2}H_0^T,\\
	\Xi_2 &= (\psi_{k+1}+v_{k+1})\zeta_k^TU_k^{-1}\zeta_k \zeta_k^TU_k^{1/2}H_0^T,\\
	\Xi_3 & = (\psi_{k+1} +v_{k+1})\zeta_k^TU_k^{-1} \zeta_k (\psi_{k+1}+v_{k+1})^T.
\end{align*}

Now by Lemma~\ref{lemma:Isserlis}, we can prove that
\begin{align*}
	\mathbb E(\Xi_1|\mathcal F_k) &= H_0U_kH_0^T\tr(U_k^{-1}) + 2H_0H_0^T,\\
	\mathbb E(\Xi_2|\mathcal F_k) &= 0 ,\\
	\mathbb E(\Xi_3|\mathcal F_k) &= \tr(U_k^{-1}) \mathbb E(\psi_{k+1}  \psi_{k+1}^T) +\tr(U_k^{-1})R.
\end{align*}
Furthermore,
\begin{align*}
	\mathbb E(\psi_{k+1}\psi_{k+1}^T) = C\Sigma C^T + \sum_{t=1}^k H_t\left(\mathbb EU_{k-t}\right)H_t^T .
\end{align*}
Now let us consider the other terms in \eqref{H0},
\begin{align*}
	&\mathbb E(H_0 U_k^{-1}\phi_ky_{k+1}^T|\mathcal F_k) = H_0H_0^T,\\
	&\mathbb E(y_{k+1}\phi_k^TU_k^{-2}H_0^T|\mathcal F_k) = H_0 H_0^T ,\\
	&\mathbb E(H_0 H_0^T|\mathcal F_k) = H_0H_0^T.
\end{align*}
Hence, the following equation holds,
\begin{align*}
	&\mathbb{E}\left(\left[y_{k+1}\phi_k^TU_k^{-1}-H_0\right]\left[y_{k+1}\phi_k^TU_k^{-1}-H_0\right]^T | \mathcal F_k\right)\\
	=&H_0U_kH_0^T\tr(U_k^{-1}) + \tr(U_k^{-1}) \psi_{k+1}  \psi_{k+1}^T +\tr(U_k^{-1})R \\
	&+ H_0H_0^T.
\end{align*}
Now if $\underline M /(k+1)^\beta\leq U_k\leq \overline M$, we can conclude that
\begin{align*}
	&\mathbb E \left( \left\|y_{k+1}\phi_k^TU_k^{-1}-H_0\right\|_F^2 \right)\\
	=&\tr\left(\mathbb{E}\left(\left[y_{k+1}\phi_k^TU_k^{-1}-H_0\right]\left[y_{k+1}\phi_k^TU_k^{-1}-H_0\right]^T \right)\right)\\
	=&O\left((k+1)^{\beta}\right).
\end{align*}
Since $\beta < 1$, according to the convergence condition of infinite series, we know that the infinite sum on LHS of \eqref{eq:boundedsum} is bounded.

Therefore, by Lemma~\ref{lemma:MatrixStrongLaw},
\begin{align*}
	\lim_{k\to\infty}\frac{S_k}{k+1} = \textbf{0}\text{ almost surely},
\end{align*}
which proves that $H_{k,0}$ converges to $H_0$ almost surely.

\section{Proof of Theorem~\ref{theorem:Pconverge}}
In order to prove theorem~\ref{theorem:Pconverge}, we need to make use of the following lemma:
\begin{lemma}
  Suppose that $\rho_k$ converges to $\rho$, where $|\rho|< 1$. Furthermore, assume that $\lim_{k\rightarrow\infty} a_k'-a_k = 0$, where $a_k$ is a bounded sequence. Then we have
  \begin{align*}
		\lim_{k\rightarrow\infty} b_k' - b_k = 0,
  \end{align*}
  where $b_k$ and $b_k'$ satisfy the following recursive equation:
  \begin{align*}
    b_{k+1} = \rho b_k + a_k,\,b'_{k+1} = \rho_k b_k' + a_k',
  \end{align*}
  with initial condition $b_{-1} = b'_{-1} = 0$.
  \label{lemma:convergeexponential}
\end{lemma}

\begin{proof}
  Notice that
  \begin{align*}
    b_{k+1}-b_{k+1}' = (\rho-\rho_k)b_k+\rho_k(b_k-b_k')  + (a_k-a_k').
  \end{align*}
Since $|\rho| < 1$ and $a_k$ is bounded. We know that $|b_k| \leq \sup_k |a_k|/(1-|\rho|)$ is also bounded. For any $\epsilon > 0$, there exists $K$, such that for any $k \geq K$, $|\rho-\rho_k|\leq \epsilon$ and $|a_k-a_k'| \leq \epsilon$. Therefore, for $k\geq K$, we have
\begin{align*}
  |b_{k+1}-b_{k+1}'| &\leq |\rho_k|\times |b_k-b_k'| + \epsilon \sup_k|b_k| + \epsilon\\
  &\leq (|\rho|+\epsilon)|b_k-b_k'| + \epsilon \sup_k|b_k| + \epsilon.
\end{align*}
Now since $|\rho|<1$, we can choose $\epsilon$ small enough such that $|\rho|+\epsilon < 1$, therefore,

\begin{align*}
  \limsup_{k\rightarrow\infty}|b_k-b_k'|\leq \frac{\epsilon}{1-|\rho|-\epsilon}(\sup_k|b_k|+1).
\end{align*}
Since $\epsilon$ can be arbitrarily small, $b_k-b_k' \to 0$.
\end{proof}

The proof of Theorem~\ref{theorem:Pconverge} is divided into 2 parts. First, by Lemma~\ref{lemma:finitetoinf}: 
\begin{align*}
	\sum_{t=0}^{k-1} H_t  U_{k-t}  H_t^T
	= \sum_{i=0}^n\sum_{j=0}^n \sum_{t=0}^{k-1} \lambda_i^t\lambda_j^t \Omega_i U_{k-t} \Omega_j^T.
\end{align*}

Therefore, by Lemma~\ref{lemma:convergeexponential}, $\mathscr U_{k,ij}$ converges to $\sum_{t=0}^{k-1} \lambda_i^t\lambda_j^t \Omega_i U_{k-t} \Omega_j^T$, which proves that
\begin{align*}
\lim_{k\rightarrow\infty}  \mathscr U_{k} - \sum_{t=0}^{k-1} H_t  U_k  H_t^T = 0\text{ almost surely}.
\end{align*}

The next step is to prove that
\begin{align*}
\lim_{k\to\infty} Y_k    - \frac{1}{k+1}\sum_{\tau = 0}^k\sum_{t=0}^{\tau-1} H_t  U_{\tau-t} H_t^T = W\text{ almost surely},
\end{align*}
which can be proved by Theorem~6 in \cite{Lyons1988} and examining the expected value of $Y_k$ and its second moment. The details are omitted due to space limit.

%\section*{Acknowledgment}
%The authors would like to thank  
\bibliographystyle{IEEEtran}
\bibliography{reference}
%\begin{IEEEbiography}{}
%	Biography text here.
%\end{IEEEbiography}
%\begin{IEEEbiographynophoto}{}
%	Biography text here.
\end{document}